\documentclass[12pt,reqno]{amsart}

\usepackage{etoolbox}
\newbool{bForSubmission}
\booltrue{bForSubmission}
\usepackage[authoryear]{natbib}

\newbool{bExperimental}

\newbool{bForUs}

\usepackage{amscd,amsmath}
\usepackage{mathrsfs}
\usepackage{amsfonts}
\usepackage{float}

\usepackage{amssymb, bm, xspace}
\usepackage{enumerate}

\usepackage{enumitem}

\usepackage[OT2,OT1]{fontenc}

\usepackage{datetime}

\usepackage{enumitem}

\usepackage[dvipsnames]{xcolor}

\usepackage[margin=1.2in]{geometry}

\usepackage{mathtools}

\mathtoolsset{centercolon}

\usepackage{lingmacros}
\usepackage{tree-dvips}
\usepackage{amsmath}
\usepackage{colonequals}
\usepackage{graphicx} 
\usepackage{plain}
\usepackage[T1]{fontenc} 
\usepackage{amssymb}

\usepackage{amsthm}
\usepackage{amscd}
\usepackage{indentfirst}
\usepackage{graphicx}
\usepackage{bm}
\usepackage{caption}

\newtheorem{Example}{Example}[section]

\newtheorem{Proposition}{Proposition}[section]

\newcommand{\R}{{\mathbb R}}

\renewcommand{\[}{\left[}

\newcommand{\beq}{\begin{equation}}
	\newcommand{\eeq}{\end{equation}}
\newcommand{\beqna}{\begin{eqnarray*}}
	\newcommand{\eeqna}{\end{eqnarray*}}
\newcommand{\beqn}{\begin{equation*}}
	\newcommand{\eeqn}{\end{equation*}}
\newcommand{\bp}{\begin{proof}}
	\newcommand{\ep}{\end{proof}}
\newcommand{\bprop}{\begin{proposition}}
	\newcommand{\eprop}{\end{proposition}}
\newcommand{\bt}{\begin{theorem}}
	\newcommand{\et}{\end{theorem}}
\newcommand{\bex}{\begin{Example}}
	\newcommand{\eex}{\end{Example}}
\newcommand{\bc}{\begin{corollary}}
	\newcommand{\ec}{\end{corollary}}
\newcommand{\bcl}{\begin{claim}}
	\newcommand{\ecl}{\end{claim}}
\newcommand{\bl}{\begin{lemma}}
	\newcommand{\el}{\end{lemma}}

\usepackage{stackengine}
\stackMath

\usepackage{tcolorbox}
\usepackage{tikz}

\newbool{HaveBBM}
\booltrue{HaveBBM}

\ifbool{HaveBBM}{
	\usepackage{bbm}
    }
    {
    }

\usepackage[pagebackref=true, colorlinks=true, citecolor=blue]{hyperref}

\usepackage{cleveref}

\tcbuselibrary{skins}
\usetikzlibrary{shadings}
\tcbset{
    myimage/.style={
        enhanced,
        overlay={
            \begin{scope}[shift={([xshift=1mm, yshift=7mm]frame.north west)}]
            \end{scope}}}}

\tcbset{
    skin=enhanced,
    fonttitle=\bfseries,
    interior style={white},
    segmentation style={black,solid,opacity=0.2,line width=1pt}}

\newtcolorbox{TitledBox}[2][]{
    myimage,              
    coltitle=black,       
    colbacktitle=white,   
    title=My title,
    attach boxed title to top center={
        yshift=-3mm,
        yshifttext=-1mm},
    attach boxed title to top left={
        xshift=1cm,
        yshift=-2mm},
    boxed title style={
        size=small},
    title={#2},#1}

\newcommand{\TwoColumn}[6]%
{%
\begin{minipage}[#3]{#1\textwidth}%
#5%
\end{minipage}%
\begin{minipage}[#4]{#2\textwidth}%
#6%
\end{minipage}%
}

\newcommand{\TwoColumnTop}[4]%
{%
\begin{minipage}[t]{#1\textwidth}
#3
\end{minipage}%
\begin{minipage}[t]{#2\textwidth}
#4
\end{minipage}%
}

\newcommand{\PullMarginsIn}[1]%
{%
\begin{minipage}[c]{0.3\textwidth}%
\phantom{x}%
\end{minipage}%
\begin{minipage}[c]{0.4\textwidth}%
#1%
\end{minipage}%
\begin{minipage}[c]{0.3\textwidth}%
\phantom{x}%
\end{minipage}%
}

\newcommand{\Comment}[1]{{\color{Brown}#1}}
\newcommand{\OptionalDetails}[1]{
    \ifbool{bForSubmission}
        {%
        }%
        {\begin{quote}\Comment{\footnotesize
        \medskip
        
        \noindent\textbf{Details not for submission}: \\
        \noindent#1}
        \end{quote}
        }
    }

\newbool{arXivFormat}
\boolfalse{arXivFormat}
    
\newcommand{\IfarXivElse}[2]{
    \ifbool{arXivFormat}
        {#1}{#2}
    }

\renewcommand{\mathbf}[1]{\bm{#1} \textbf{ *** Use bm instead of mathbf ***}}

\newcommand{\eqn}{\begin{eqnarray}}
\newcommand{\een}{\end{eqnarray}}

\newtheorem{theorem}{Theorem}[section]
\newtheorem*{theorem*}{Theorem}				

\newtheorem{lemma}[theorem]{Lemma}

\newtheorem{remark}[theorem]{Remark}
\newtheorem*{remark*}{Remark}

\numberwithin{equation}{section}

\newcommand{\CharFunc}{
    \ifbool{HaveBBM}{
        {\ensuremath{\mathbbm{1}}}
        }
        {
        {\ensuremath{\bm{1}}}
        }
    }

\newcommand\tenq[2][1]{%
	\def\useanchorwidth{T}%
	\ifnum#1>1%
		\stackunder[0pt]{\tenq[\numexpr#1-1\relax]{#2}}{\scriptscriptstyle\sim}%
	\else%
		\stackunder[1pt]{#2}{\scriptscriptstyle\sim}%
	\fi%
	}

\renewcommand{\epsilon}{\varepsilon}
\newcommand{\Ignore}[1]{}

\usepackage{todonotes}
\newcommand{\Obsolete}[1]{}

\newcommand{\Experimental}[1]%
{%
\ifbool{bExperimental}%
	{\bigskip
	\noindent
	\textbf{\color{Brown}*** Start experimental}\\
	#1
}
{
}
}

\definecolor{Correction}{named}{red}

\crefname{cor}{Corollary}{Corollaries} 
									   
\crefname{lemma}{Lemma}{Lemmas}	       

\crefname{section}{Section}{Sections}
\Crefname{section}{Section}{Sections}

\crefname{appendix}{Appendix}{Appendices}
\Crefname{appendix}{Appendix}{Appendices}

\crefname{theorem}{Theorem}{Theorems}
\Crefname{theorem}{Theorem}{Theorems}

\crefname{prop}{Proposition}{Propositions}
\Crefname{prop}{Proposition}{Propositions}

\crefname{conj}{Conjecture}{Conjectures}
\Crefname{conj}{Conjecture}{Conjectures}

\crefname{definition}{Definition}{Definitions}
\Crefname{definition}{Definition}{Definitions}

\crefname{remark}{Remark}{Remarks}
\Crefname{remark}{Remark}{Remarks}

\crefname{assumption}{Assumption}{Assumptions}
\Crefname{assumption}{Assumption}{Assumptions}

\crefformat{equation}{(#2#1#3)}
\crefrangeformat{equation}{(#3#1#4) through (#5#2#6)}
\crefmultiformat{equation}
    {(#2#1#3)}%
    { and~(#2#1#3)}
    {, (#2#1#3)}
    { and~(#2#1#3)}

\def\XXint#1#2#3{{\setbox0=\hbox{$#1{#2#3}{\int}$ }
\vcenter{\hbox{$#2#3$ }}\kern-.6\wd0}}

\allowdisplaybreaks

\begin{document}
\newdateformat{mydate}{\THEDAY~\monthname~\THEYEAR}

\title
	[Model of an enzyme]
	{Enzyme-Substrate Complex Formation Modulates Diffusion-Driven Patterning in Metabolic Pathways }


\author[F. Farivar]
{Faezeh Farivar$^{1}$}
\address{$^1$  Institute of Atmospheric Sciences and Climate (CNR-ISAC), Rome, Italy} 
\email{faezehfarivar@cnr.it}


\begin{abstract}
Spatial organization in metabolic pathways can arise from the interplay between enzymatic reaction kinetics and diffusion-driven instabilities. In this work we investigate how reversible enzyme--substrate binding influences pattern formation in a two-step metabolic pathway. Starting from a mechanistic description in which the substrate reversibly binds to the first enzyme before catalytic conversion, we formulate a three-species reaction--diffusion system that explicitly incorporates the enzyme--substrate complex.

We first analyse the homogeneous dynamics and determine the unique steady state of the kinetic system. Exploiting the separation of time scales between the rapid binding kinetics and the slower evolution of metabolite concentrations, we derive a reduced two-variable model using a quasi-steady-state approximation for the enzyme--substrate complex. This reduction preserves the essential nonlinear coupling between catalytic reactions and spatial transport.

Linear stability and weakly nonlinear analysis reveals conditions for diffusion-driven (Turing) instability and shows that reversible enzyme binding significantly modifies the location and extent of the instability region compared to models with effective kinetics. Numerical simulations confirm the analytical predictions and demonstrate how enzyme--substrate interactions reshape pattern selection and slow the emergence of spatial heterogeneity.

These results provide a mechanistic link between enzyme binding kinetics, diffusion-driven pattern formation, and mesoscale metabolic organization. The proposed framework offers a tractable approach for studying spatial patterning in enzymatic networks and may help explain the emergence of structured biochemical domains such as those associated with liquid--liquid phase separation.
\end{abstract}

\maketitle

\markleft{F. FARIVAR}

\vspace{-2.5em}

\begin{center}
\medskip
Compiled on {\dayofweekname{\day}{\month}{\year} \mydate\today} at \currenttime
		
\end{center}

\bigskip

\vspace{-2.0em}

\setcounter{tocdepth}{1}
{
\renewcommand\contentsname{}	

\tableofcontents
}

\normalsize

%


%

\section{Introduction}

Understanding the complex interactions in biochemical reaction networks is essential for uncovering the principles of cellular function and metabolic regulation. Among the key features of these systems are reaction-diffusion processes, in which chemical species diffuse through space and interact in nonlinear ways. These interactions can give rise to striking spatial organizations, most notably Turing patterns, which are periodic structures that emerge spontaneously from homogeneous initial conditions when reaction kinetics and diffusion rates satisfy certain instability criteria. Since the pioneering work of Turing (\cite{turing1952chemical}), reaction-diffusion mechanisms provide a robust mathematical framework for the
generation of spatial patterns in biological systems (\cite{maini2012turing, maini2019turing}).
 ranging from chemical oscillators   (\cite{epstein2012chemical}, \cite{gambino2015effects}) to developmental patterning in multicellular systems (\cite{furusawa1998emergence, giunta2021pattern}).  

Metabolic pathways in living organisms provide a particularly rich setting for studying such processes. These pathways consist of sequences of enzymatic reactions that convert substrates into intermediates and final products, often with multiple branches and feedback loops (\cite{hafner2021nicepath}, \cite{lee2025biochemical}, \cite{hill2015metabolomics}). Spatial aspects of metabolism have received growing attention in recent years, as it has become clear that enzymes and metabolites are not always homogeneously distributed in the cytoplasm. Instead, cells frequently organize enzymes into clusters or into \emph{biomolecular condensates} through liquid--liquid phase separation (LLPS) (\cite{milicevic2022emerging}, \cite{li2025peptide}, \cite{chen2022liquid}). Such mesoscale organization can accelerate metabolic fluxes, buffer fluctuations, and regulate competition at branch points (\cite{bevilacqua2024enzyme}, \cite{lu2021emerging}, \cite{lim2024phase}, \cite{krainer2021reentrant, liu2023liquid}).
Liquid‒liquid phase separation is a fundamental
mechanism of sepsis (\cite{chen2025liquid}).
 From a theoretical perspective, reaction-diffusion models offer a natural framework to investigate how enzyme clustering and LLPS influence the spatial dynamics of metabolism (\cite{castellana2014enzyme,  mukherjee2024reaction}, \cite{buchner2013clustering, cieza2022investigation, nakasone2021time, rigano2022models, kim2024reaction}).  For instance, liquid-liquid phase separation (LLPS) on cell membranes has been described using reaction-diffusion models, which demonstrate that enzymatic feedback loops may drive out-of-equilibrium dynamics that replicate classical phase separation properties as nucleation, coarsening, and domain creation \cite{rigano2022models}. A straightforward reaction-diffusion framework called the swarm model has been put out to explain the molecular condensation processes that occur in cells.  It demonstrates how these condensates emerge out of equilibrium and may be controlled by nucleation and cooperativity dynamics. It captures important aspects of membraneless organelle development, including nucleation, clustering, diffusion, and particle exchange (\cite{cieza2022investigation}).
 In this work, we introduce Models~1, 2, and~3 to investigate the conditions for Turing pattern formation in enzyme-mediated metabolic pathways.
The Table \ref{tab:par-model1-2} summarizing all parameters contained in Models 1 and 2 followed as:
\begin{table}[H]
\begin{small}
\centering
\begin{tabular}{ll}
\hline
\textbf{Parameter} & \textbf{Description (with units)} \\
\hline
$C_0, C_1$ & Substrate and intermediate concentrations \; [$\mu$M] \\
$C_{E_1S_0}$ & Enzyme--substrate complex concentration \; [$\mu$M] \\[4pt]
$n_1, n_2$ & Effective enzyme levels (activity or concentration) \; [dimensionless] \\[4pt]
$\alpha_0$ & Relaxation rate toward $C_0^*$ \; [s$^{-1}$] \\
$C_0^*$ & Reference substrate level \; [$\mu$M] \\[4pt]
$k_1, k_2$ & Catalytic prefactors for $E_1$, $E_2$ \; [$\mu$M$^{-1}$ s$^{-1}$] \\[4pt]
$k_{a_1}$ & Association (binding) rate of $E_1$ to $S_0$ \; [$s^{-1}$] \\
$k_{d_1}$ & Dissociation rate of $E_1S_0$ \; [s$^{-1}$] \\
$k_{cat}$ & Catalytic turnover rate to $S_1$ \; [s$^{-1}$] \\[4pt]
$\beta$ & Degradation rate of $C_1$ \; [s$^{-1}$] \\[6pt]
$D_0, D_1, \gamma_0, \gamma_1$ &
Diffusion coefficients of $C_0$ and $C_1$ \; [$\mu$m$^2$/s] \\
$\gamma_{\mathrm{comp}}$ &
Diffusion coefficient of enzyme-substrate complex \; [$\mu$m$^2$/s] \\[4pt]
$d,\ \gamma$ & Cross-diffusion coupling strengths \; [$\mu$m$^2/(\mu M s)$] \\
\hline
\end{tabular}
\caption{\textit{Summary of parameters used in Models~1 and~2 with physical units.}}
\label{tab:par-model1-2}
\end{small}
\end{table}
Unless otherwise stated, all parameters appearing in the models and numerical simulations use the dimensional units listed in Table~\ref{tab:par-model1-2}.

The novelty of this work lies in combining mechanistic enzyme–substrate kinetics
with diffusion and cross-diffusion to study the emergence of spatial metabolic
organization. We introduce a reaction-diffusion formulation that
explicitly incorporates reversible enzyme-substrate complex formation and derive a
reduced model via a quasi-steady state approximation (QSSA) that preserves the essential
nonlinear feedback between catalysis and spatial transport. By comparing this model
with a classical simplified enzymatic pathway, we demonstrate how complex formation
shifts the homogeneous steady state, alters relaxation dynamics, and significantly
modifies the Turing instability region. To our knowledge, this is the first analysis
quantifying how enzymatic binding kinetics reshape diffusion-driven instabilities in
two-step metabolic pathways, providing a mechanistic link between enzyme-mediated
feedback and mesoscale organization such as liquid–liquid phase separation.
\subsection{\textbf{Model 1: The primary framework.}}  
As a starting point, we consider a simplified two-step metabolic pathway
(\cite{castellana2014enzyme}):
$
S_0 \xrightarrow{E_1} S_1 \xrightarrow{E_2} P,
$
where  $S_0$, $S_1$, and $P$ denote substrate, intermediate, and final
product, respectively. The substrate $S_0$ is converted to the
intermediate $S_1$ by enzyme $E_1$, which in turn is converted to $P$
by enzyme $E_2$. Following the phenomenological formulation of
(\cite{castellana2014enzyme}), we model the
spatio-temporal evolution of the substrate and intermediate as
\begin{equation}\label{primary model}
\begin{cases}
\dfrac{\partial C_0}{\partial t}
   =-\alpha_0(C_0-C_0^*)-k_1n_1C_0
     +D_0\nabla^2C_0 - d\,C_0\nabla^2C_1,\\[6pt]
\dfrac{\partial C_1}{\partial t}
   =-k_2n_2C_1-\beta C_1
     +k_1n_1C_0 + D_1\nabla^2C_1.
\end{cases}
\end{equation}
Here $C_0$ and $C_1$ denote the concentrations of $S_0$ and $S_1$, while
$n_1$ and $n_2$ represent fixed enzyme levels. The parameter
$\alpha_0$ drives $C_0$ toward its homeostatic value $C_0^*$,
$k_1n_1$ and $k_2n_2$ are effective catalytic and consumption rates,
and $\beta$ is the degradation rate of $C_1$. For $\beta=0$, the
intermediate can accumulate indefinitely, eliminating steady-state
flux through the pathway.

\medskip
\textbf{\textit{Curvature-driven cross-diffusion.}}
The cross–diffusion term appearing in the equation for $C_0$ models the influence of spatial variations of the intermediate metabolite $C_1$ on the transport of the substrate. Such effects arise naturally when the flux of one species depends on gradients of another species.

Specifically, we assume that the flux of the substrate has the form
$$
J_0 = -D_0 \nabla C_0 + d C_0 \nabla C_1,
$$
where the first term represents standard Fickian diffusion and the second term represents a drift induced by spatial gradients of the intermediate metabolite.

Applying the conservation law $\partial_t C_0 = -\nabla \cdot J_0$ yields
$$
\partial_t C_0
= D_0 \nabla^2 C_0 - d \nabla \cdot (C_0 \nabla C_1).
$$

For analytical simplicity, we approximate the divergence term by retaining the dominant curvature contribution, leading to the cross–diffusion term $- C_0 \nabla^2 C_1$ used in system~\eqref{Modified model}. This formulation captures the influence of spatial variations of $C_1$ on the transport of $C_0$ while preserving a tractable mathematical structure. Cross-diffusion mechanisms of this type have been extensively studied in
biological pattern-formation models, where they represent directed transport
induced by concentration gradients or local interactions (\cite{hillen2009user}).

This coupling does not depend on the gradient of $C_1$ but on its
\emph{curvature}: at locations where $C_1$ attains a local maximum,
$\nabla^{2}C_1<0$ and therefore $-d\,C_0\nabla^2 C_1>0$, implying that
$C_0$ increases at such sites. Conversely, $C_0$ decreases near local minima
($\nabla^{2}C_1>0$). Thus, substrate accumulates preferentially near
regions where the intermediate forms concave peaks, producing an effective
recruitment mechanism without invoking direct molecular attraction.

Such curvature-dependent transport is consistent with mesoscale
biophysical processes. Curvature-sensing proteins in membranes preferentially
bind to negatively curved regions (\cite{arnold2025bending}), and curvature
gradients can drive molecular fluxes in cells and tissues
(\cite{schamberger2023curvature}). Analogous behavior occurs in LLPS
condensates and metabolon-like assemblies, where spatial variations in
viscosity, binding affinity, and chemical potential bias metabolite
transport (\cite{castellana2014enzyme,zwicker2017growth}). The term
$-d\,C_0\nabla^{2}C_1$ therefore provides a minimal and biophysically
motivated representation of curvature-driven substrate recruitment.

LLPS proceeds through nucleation, maturation, and functional biochemical
activity (\cite{hyman2014liquid, banani2017biomolecular}). Model~1 corresponds to the latter two stages: we assume that a
condensate has already formed and maintains approximately constant
internal enzyme concentrations $n_1$ and $n_2$, while the metabolites
$C_0$ and $C_1$ remain mobile and diffuse into, out of, and within the
condensed phase. The curvature-driven term $-d\,C_0\nabla^2 C_1$ captures
substrate recruitment into regions of high intermediate accumulation,
a hallmark of the mature, functional stage of LLPS-mediated metabolic
organization.

Model~1 thus serves as a coarse-grained baseline for studying generic
reaction–diffusion feedbacks in enzyme-containing LLPS compartments and
provides the foundation upon which mechanistic extensions (Model~2) can
be built.

\noindent
\subsection{\textbf{Model 2: Explicit enzyme–substrate interactions.}}
Model~1 captures basic spatial coupling but treats catalysis as an effective one-step process and does not resolve the transient enzyme–substrate complex. In cellular environments, however, enzymes bind substrates reversibly before turnover (\cite{keener2009mathematical}), and the mobility of such complexes—especially inside co-localized assemblies such as liquid–liquid phase-separated (LLPS) condensates—may differ substantially from that of free metabolites. 

To account for these effects, we introduce an additional dynamical variable representing the enzyme–substrate complex $C_{E_1S_0}$. The underlying biochemical scheme,
$$
E_1 + S_0 \xrightleftharpoons[K_{d_1}]{k_{a_1}} E_1S_0 \xrightarrow{k_{cat}} E_1 + S_1,
$$
motivates the extended reaction–diffusion system
\begin{equation}\label{Modified model}
\begin{cases}
\dfrac{\partial C_0}{\partial t}
   =-\alpha_0(C_0-C_0^*)
     -k_{a_1}C_{E_1S_0}
     +\gamma_0\nabla^2C_0
     -\gamma\ C_0\nabla^2C_1,\\[6pt]
\dfrac{\partial C_{E_1S_0}}{\partial t}
   =k_{a_1}n_1 C_{0}
     -(k_{d_1}+k_{cat})C_{E_1S_0}
     +\gamma_{\mathrm{comp}}\nabla^2C_{E_1S_0},\\[6pt]
\dfrac{\partial C_1}{\partial t}
   =-(k_2n_2+\beta)\,C_1
     +k_{cat}C_{E_1S_0}
     +\gamma_1\nabla^2 C_1,
\end{cases}
\end{equation}

\textbf{Remark.}
Equation~\eqref{Modified model} is not intended as a strict microscopic mass-action description. 
Instead, $C_{E_1S_0}$ acts as an \emph{effective complex-density field} summarizing local co-localization of enzyme and substrate—for example within LLPS condensates or metabolon clusters. 
The terms $k_{a_1}C_0$ and $k_{d_1}C_{E_1S_0}$ therefore represent coarse-grained recruitment and release processes characteristic of mesoscale organization (\cite{castellana2014enzyme,zwicker2017growth}).

Model~2 extends the phenomenological dynamics of Model~1 by introducing an explicit intermediate while preserving the same basic mechanism: feedback between catalytic activity and spatial metabolite organization. 
Model~1 provides a minimal two-species description, whereas Model~2 offers a more mechanistic formulation linking enzyme binding, catalytic turnover, and condensate-mediated localization.

\medskip
\noindent
\medskip
\noindent
\subsection{\textbf{Model 3: Reduced  formulation.}}

While explicitly including the enzyme--substrate complex enriches the model
biologically, it also increases its dimensionality and analytical complexity.
In particular, the three--species formulation introduces an additional
timescale associated with the rapid formation and dissociation of the
enzyme--substrate complex.

To enable direct comparison with the primary two-variable pathway and to
facilitate analytical progress, we derive a reduced two-variable model
(Section~\ref{Sec:non-deg-model}) by eliminating the intermediate complex
through a quasi-steady-state approximation (QSSA). This reduction is
motivated by the assumption that the enzyme-substrate complex relaxes
rapidly to equilibrium relative to the slower evolution of the substrate
concentrations.

Under this approximation, the complex concentration is expressed
algebraically in terms of the substrate, leading to a closed system for the
substrate and intermediate species. The resulting reduced model preserves
the mechanistic effects of enzyme binding and catalytic conversion through
renormalized reaction coefficients, while retaining the same
cross--diffusion structure that captures spatial coupling effects.

This formulation provides a simplified yet mechanistically consistent
description of the pathway dynamics, enabling direct comparison with the
primary model and supporting the subsequent analysis of diffusion--driven
instabilities.

Throughout this work we impose for all models and all parameters are primary considered non-negative and  the homogeneous Neumann (zero–flux) boundary
conditions.
These conditions ensure that mass cannot be lost through the boundary and reflect the fact
that metabolites and enzyme–substrate complexes do not cross the boundary of the system.

Although the numerical simulations presented in Section~\ref{sec:num-results} are performed in one spatial
dimension, the reaction–diffusion mechanism is not restricted to 1D. The same
chemical interactions and diffusion laws apply in two and three dimensions, and the
Turing conditions are dimension-independent: what changes is only the geometry of
the domain and the admissible wave numbers. In higher dimensions, additional
pattern morphologies (spots, stripes, labyrinths, hexagons)(\cite{murray2007mathematical}, \cite{gambino2013pattern, gambino2019pattern}) may occur, but the
biochemical interpretation—enzymatic feedback coupled with substrate/intermediate
mobility—remains unchanged. Thus, the 1D setting captures the essential
instability mechanism while allowing clearer analytical and numerical comparison
between the simplified and reduced models.

The remainder of this paper is organized as follows. In Section~\ref{sec:primarymodel}, we analyze the local dynamics of the primary model and derive the conditions for Turing instability driven by diffusion and cross-diffusion effects. In Section~\ref{sec:model2-3}, we examine the equilibrium structure of the extended three-variable system, showing that it is simplified. To address this,  we employ a quasi-steady state approximation (\cite{keener2009mathematical, murray2007mathematical}) to obtain a reduced  model, for which we investigate both local stability and the onset of Turing instability. Section~\ref{sec:turinregion} explores how variations in reaction and diffusion parameters affect the extent of Turing regions in Models~1 and~3. In Section~\ref{sec:wnl}, we apply weakly nonlinear (WNL) analysis to the primary and reduced models to derive amplitude equations describing the nonlinear evolution of spatial patterns. These analytical predictions are then compared with numerical simulations in Section~\ref{sec:num-results}, confirming the validity of the theoretical results. Finally, Section~\ref{sec:conclu} presents concluding remarks and discusses the broader implications of our findings for understanding spatial regulation in enzyme-mediated and phase-separated biochemical systems.

\section{Simplified Two-Step Enzyme Model}\label{sec:primarymodel}
 In this section, we present linear analysis of the primary model \eqref{primary model} around the unique steady state of the system. For all positive non-negative parameters, this system has an unique coexistence steady state $E^*=(C_0^E, C_1^E)$ that is:
 \begin{equation}\label{equlibrium}
  C_0^E=\dfrac{\alpha_0 C_0^*}{\alpha_0+k_1n_1}, \;\; C_1^E=\dfrac{k_1n_1\alpha_0C_0^*}{(k_2n_2+\beta)(\alpha_0+k_1n_1)},
 \end{equation}

 Linear analysis of the kinetic around $E^*$ demonstrates that  $E^*$ is always stable  since
 \begin{equation}\label{reac-prim}
 J(E^*)=\begin{pmatrix}
 -\alpha_0-k_1n_1 && 0\\
 k_1n_1 && -\beta-k_2n_2
 \end{pmatrix},
 \end{equation}
 So the characteristic function is obtained for 
 $ \vert \lambda I- J(E^*)\vert= \lambda^2+\lambda tr(J)+det(J)=0,$ that has two negative eigenvalues $\lambda_1=-\alpha_0-k_1n_1,\; \; \lambda_2=-\beta-k_2n_2.$ \\
 \begin{Proposition} The system \eqref{primary model} does not go under \textit{Hopf bifurcation }since all parameters are non negative and $tr(J)<0$. 
 \end{Proposition}
 \subsection{Turing Analysis}
 According to the Turing's statement, a steady state can be unstable in presence of diffusion. In order to investigate Turing instabilities, one needs to linearize whole system at $E^*$, which provides:
 \begin{equation}\label{diffusion-prim}
 \mathcal{D}(E^*)=
 \begin{pmatrix}
 D_0 && -dC_0^E\\
 0 && D_1
  \end{pmatrix},
  \end{equation}
  and for $W=[C_0-C_0^E, C_1-C_1^E] $, we have $$ \dot{W}=J(E^*)W+\mathcal{D}(E^*)\nabla^2 W,$$
In order to find the solution of the linearized model with replacing the system with $W=e^{ikx+\lambda t}$, in which $k$ and $\lambda$ imply wave number and growth rate correspondingly.
Therefore, dispersion relation of the model is determined as $\vert\lambda I-J+k^2 D\vert =0$.
Thus, the dispersion relation is given by: 
\begin{equation}
\lambda^2+A(k^2)\lambda+B(k^2)=0,
\end{equation}
\begin{equation}
A(k^2)=k^2(D_1+D_0)+\alpha_0+k_1n_1+\beta+k_2n_2,
\end{equation}
  According to the Turing analysis, $E^*$ is unstable as one eigenvalue of the dispersion relation is positive, that requires $B(k^2)<0$ for a range of wave numbers $[k_a,k_b]$, where 
  \begin{eqnarray}
  B(k^2)&=  k^4 D_0D_1+k^2[-D_0(-\beta-k_2n_2)-D_1(-\alpha_0-k_1n_1)-k_1n_1dC_0^E]\\
  &+(-k_1n_1-\alpha_0)(-\beta-k_2n_2)=b_1k^4+b_2k^2+b_3,
  \end{eqnarray}
Turing necessary condition obtains that in the Turing threshold 
\begin{align}\label{B_min}
B(k_{min}^2)=0,
\end{align}
 which gives
\begin{eqnarray}
k^2_{min}=\dfrac{-b_2}{2b_1},
\end{eqnarray}  
Since only positive wave numbers are physically meaningful and $b_1>0$, therefore, $b_2$ must be negative. This provides that critical Turing bifurcation parameter is $d$ and it is determined by:
\begin{equation}\label{Turing cindition 1}
 d>\dfrac{-D_0(-k_2n_2-\beta)-D_1(-k_1n_1-\alpha_0)}{k_1n_1C0_E^*},
\end{equation}
So to find the critical wave number we replace
$d_c=\eta_1/\eta_2+\varepsilon$ where 
\begin{eqnarray}
&\eta_1=-D_0(-k_2n_2-\beta)-D_1(-k_1n_1-\alpha_0),\\
&\eta_2=k_1n_1C_0^E,
\end{eqnarray}
  into \eqref{B_min} which obtains that 
  \begin{equation}
  \varepsilon=\dfrac{2\sqrt{D_0D_1 (-\alpha_0-k_1n_1)(-\beta-k_2n_2)}}{k_1n_1 C0_E^*},
  \end{equation}
and consequently 
\begin{equation}
k_{min}^2=\sqrt{\dfrac{(-\alpha_0-k_1n_1)(-\beta-k_2n_2)}{D_0D_1}}.
\end{equation}
Hence, the system goes to the pattern formation if necessary condition \eqref{Turing cindition 1} is satisfied.

\begin{remark} It is easy to check that in absence of cross-diffusion $(d=0)$, they system remain stable since $b_2>0$ for any given parameters. So, cross-diffusion term $d$ is only key mechanism of emergence of Turing patterns. 
\end{remark}
  \section{Modified Model and Quasi-Steady state Approximation }\label{sec:model2-3}
To obtain a more mechanistic representation of the enzymatic conversion process, the reaction scheme is extended to include the reversible formation of the enzyme–substrate complex and its catalytic conversion into the intermediate. This leads to the three–species reaction–diffusion system \eqref{Modified model}. We first analyse the spatially homogeneous dynamics obtained by neglecting diffusion terms.
\subsection{Equilibrium Analysis and Derivation of the Reduced Model}
\label{Sec:non-deg-model}

To determine the homogeneous equilibrium concentrations
$C_0^e$, $C_1^e$, and $C_{E_1S_0}^e$, we consider the steady-state
conditions associated with system \eqref{Modified model}, obtained by
setting all time derivatives to zero and neglecting spatial variations.

The steady-state equations are therefore
\begin{equation}
\begin{cases}
0=-\alpha_0(C_0-C_0^*)-k_{a_1} C_{E_1S_0},\\
0=k_{a_1}n_1 C_0-(k_{d_1}+k_{cat})C_{E_1S_0},\\
0=-(k_2n_2+\beta) C_1+k_{cat} C_{E_1S_0}.
\end{cases}
\end{equation}

From the second equation we obtain the relation
$
C_{E_1S_0}
=\dfrac{k_{a_1}n_1}{k_{d_1}+k_{cat}} C_0 .
$

Substituting this expression into the third equation gives
$
C_1=
\dfrac{k_{cat}}{k_2n_2+\beta}C_{E_1S_0}.
$

Finally, inserting the expression for $C_{E_1S_0}$ into the first equation yields the equilibrium value of the substrate concentration
$$
C_0^e=
\frac{\alpha_0 C_0^*}
{\alpha_0+\dfrac{k_{a_1}^2n_1}{k_{d_1}+k_{cat}}}.
$$

Consequently the homogeneous equilibrium of the full system is
$
E^*=(C_0^e,C_{E_1S_0}^e C_1^e),
$
where
$
C_{E_1S_0}^e=
\dfrac{k_{a_1}n_1}{k_{d_1}+k_{cat}} C_0^e,
\qquad
C_1^e=
\dfrac{k_{cat}}{k_2n_2+\beta}  C_{E_1S_0}^e.
$

\medskip
\noindent
\textbf{Quasi-steady-state reduction.}
In many enzymatic pathways the formation and dissociation of the
enzyme--substrate complex occur on a much faster time scale than the
evolution of the substrate concentrations. Consequently, the complex
rapidly relaxes toward a local quasi-steady state determined by the
balance of the binding and catalytic reactions. Under this separation
of time scales, the diffusive contribution to the equation governing
$C_{E_1S_0}$ can be neglected to leading order.

Applying this quasi-steady-state approximation to the second equation of
system \eqref{Modified model} gives the algebraic relation
$
(k_{d_1}+k_{cat}),C_{E_1S_0}\approx k_{a_1}n_1 C_0,
$
from which
\begin{equation}
\label{CES}
C_{E_1S_0}\approx
\frac{k_{a_1}n_1}{k_{d_1}+k_{cat}}C_0.
\end{equation}

Substituting the approximation \eqref{CES} into the remaining equations of
system \eqref{Modified model} yields the reduced two-variable
reaction-diffusion system
\begin{equation}
\label{non-degenerated}
\begin{cases}
\dfrac{\partial u_0}{\partial t}
= -\alpha_0 \big(u_0 - u_{in}^*\big)

-\dfrac{k_{a_1} ^2 n_1}{k_{d_1} + k_{cat}} u_0

+ \gamma_0 \nabla^2 u_0

- \gamma u_0 \nabla^2 u_1, \\
  \dfrac{\partial u_1}{\partial t}
  = - (k_2 n_2 + \beta) u_1

+ \dfrac{k_{a_1} k_{cat} n_1}{k_{d_1} + k_{cat}} u_0
+ \gamma_1 \nabla^2 u_1 .
  \end{cases}
  \end{equation}

The reduced system \eqref{non-degenerated} retains the same structural
form as the primary reaction--diffusion model \eqref{primary model},
but with effective reaction coefficients that incorporate the influence
of the enzyme--substrate complex through the quasi-steady-state
approximation \eqref{CES}. In this formulation the complex no longer
appears as an independent dynamical variable, while the biochemical
effects of enzyme binding and catalytic conversion remain embedded in
the renormalized reaction terms.

This reduction provides a simplified description of the pathway dynamics
that facilitates analytical comparison with the primary model, while the
full three-species system is retained in the subsequent analysis of
diffusion-driven pattern formation.

\subsection{{Linear Analysis of the  Reduced Model}}
The reduced  system \eqref{non-degenerated}, obtained via the
quasi-steady state approximation of the enzyme-substrate complex, admits a
unique positive coexistence steady state given by
\begin{equation}\label{modified-Equ}
u_0^*=\dfrac{\alpha_0u_{in}^*(k_{d_1}+k_{cat})}
{\alpha_0(k_{d_1}+k_{cat})+k_{a_1}^2 n_1},\qquad
u_1^*=\dfrac{k_{a_1}k_{cat}n_1 u_{0}^*} 
{(k_{d_1}+k_{cat})(k_{2}n_2+\beta)}.
\end{equation}
All parameters are assumed to be positive.

This model contains all positive parameters. Linearizing the model at $E^*=(u_0^*, u_1^*)$ gives 
\begin{equation}
 J(E^*)=\begin{pmatrix}
 -\alpha_0-\dfrac{k_{a_1}^2n_1}{k_{d_1}+k_{cat}} && 0\\
 \dfrac{k_{a_1}k_{cat}n_1}{k_{d_1}+k_{cat}} && -\beta-k_2n_2
 \end{pmatrix},\;\; \; \mathcal{D}(E^*)=\begin{pmatrix}
 \gamma_0  && -\gamma u_0^*\\
 0 && \gamma_1
 \end{pmatrix}.
\end{equation}
This system is locally stable at $E^*$ simultaneously. And it admits Turing instability in presence of diffusion with necessary condition
\begin{equation}\label{turing condition-M_model}
 \gamma >\dfrac{-\gamma_0(-k_2n_2-\beta)-\gamma_1(-\alpha_0-\dfrac{k_{a_1}^2n_1}{k_{d_1}+k_{cat}})}{\dfrac{k_{a_1}k_{cat}n_1}{k_{d_1}+k_{cat}}u_0^*},
 \end{equation}
 So $\gamma$ is  Turing bifurcation parameter.
 Moreover, critical wave number is respectively determined
 $$ k_{min}^2=\sqrt{\dfrac{(\alpha_0-\dfrac{k_{a_1}k_{d_1}n_1^2}{k_{d_1}+k_{cat}})(-k_2n_2-\beta)}{\gamma_0 \gamma_1}},$$
\begin{figure}[h]
\centering
\vspace{-5cm}
\includegraphics[width=0.9\textwidth]{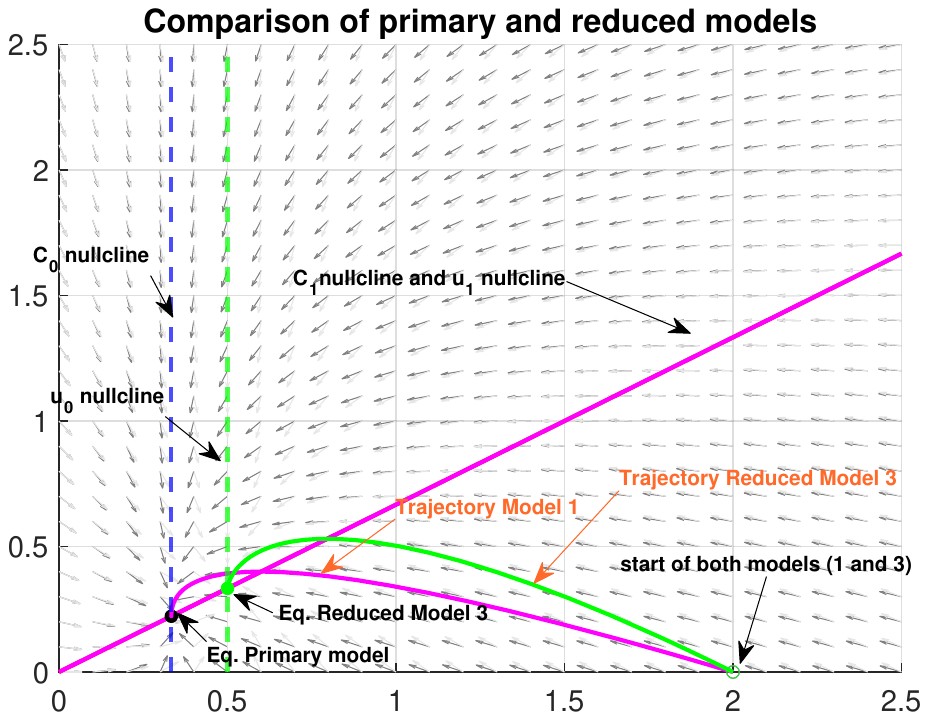}
\vspace{-5cm}
\caption{\textit{\textbf{Phase-plane comparison of the simplified  enzyme-pathway models.}
Arrows represent vector fields (dark gray: primary model; light gray: simplified model). 
Red/blue and magenta/green curves denote the corresponding nullclines, whose intersections yield homogeneous steady states (dots). 
Trajectories from identical initial conditions are shown in magenta and green. 
The simplified formulation shifts the steady state toward higher intermediate concentration and exhibits slower relaxation, 
consistent with nonlinear feedback between catalysis and compartmentalization.}}
\label{fig:nullclines}
\end{figure}
We have investigated numerically necessary conditions of appearance of Turing instabilities of both models \eqref{primary model} and \eqref{non-degenerated} of a given set of parameters (see (Figure~\ref{fig:B and Lambda plots})). In this figure, plot of $B(k^2)$ demonstrate as bifurcation parameters $d$ of \eqref{primary model} and $\gamma$ of \eqref{non-degenerated} crosses their critical values, necessary conditions \eqref{Turing cindition 1} and \eqref{turing condition-M_model} are satisfied and correspondingly for both models $B(k^2)$ admit negative value and respectively $\lambda(k^2)$ admit positive values, which states appearance of pattern formation.

 \begin{figure}[h]
		\centering
		{\includegraphics[width=0.49\textwidth]{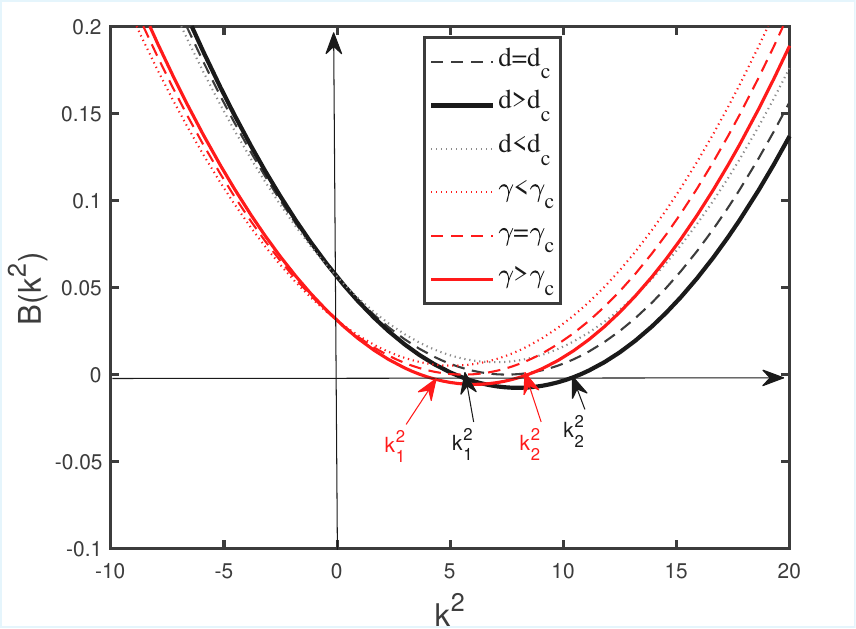}}
		{\includegraphics[width=.455\textwidth]{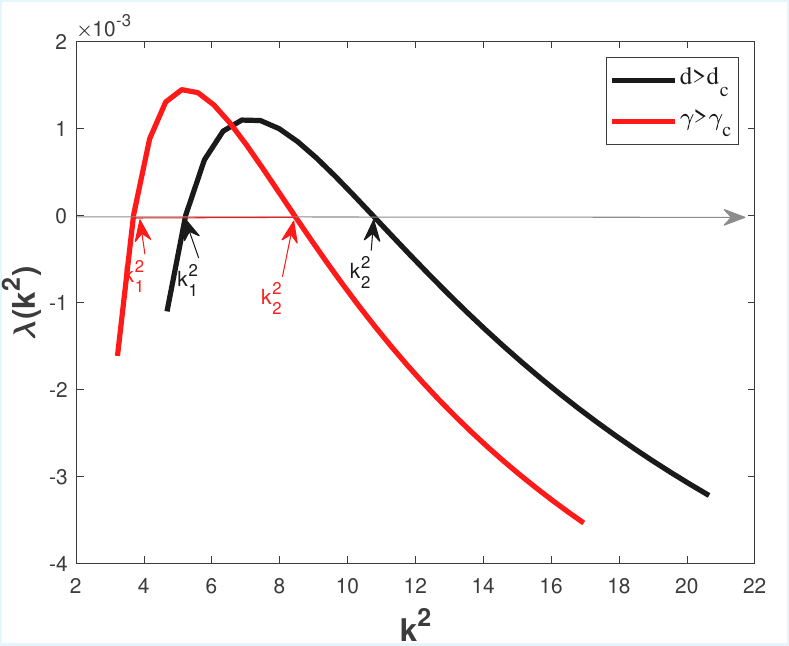}}		
		\caption{\textit{Plots of polynomials $B(k^2)$ of two models \eqref{non-degenerated} (red plots) and \eqref{primary model} (black plots) (right figure). Plots of dispersion relation polynomials $\lambda^2(k^2)$ (left figure). Given parameters are $\alpha_0=0.01 $, $k_1=0.05 $, $k_2=0.02 $, $k_{d_1}=0.01 $, $k_{a_1}=0.03 $, $k_{cat}=0.1 $, $\beta=0.01$, $C_0^*=0.1 $, $n_1=10 $, $n_2=5 $, $\gamma_0=0.1 $, $\gamma_1=0.01 $.  where in this figures $d_c$ and $\gamma_c$ are Turing bifurcation parameter of model \eqref{primary model} and \eqref{non-degenerated} respectively.}}
		\label{fig:B and Lambda plots}
		\hspace{1mm}
	\end{figure}
A qualitative comparison between the kinetic dynamics of the simplified model (1.1) 
and the reduced  model (3.3) is shown in Figure~\ref{fig:nullclines}. 
The phase–plane analysis highlights several key differences: (i) the nullclines of the 
reduced model intersect at a steady state with higher intermediate concentration, 
(ii) trajectories relax more slowly toward equilibrium, reflecting the additional nonlinear 
feedback introduced by enzyme–substrate binding, and (iii) the vector fields of the two 
systems display distinct curvature, indicating that the mechanistic correction introduced 
by the approximation reduction significantly alters the underlying kinetics. 	
\section{Turing region}\label{sec:turinregion}
To understand how parameters influence the emergence of Turing instability
in the full enzymatic model, and to compare these effects with those
predicted by the reduced  system, we investigate the
corresponding Turing regions in parameter space.
 In order to achieve this, we need to consider positivity conditions of both $E^*$, stability conditions in absence  of diffusion. According to the \eqref{equlibrium} and \eqref{modified-Equ} for the parameter ranges considered in this study, the coexistence steady
states of both models are positive and locally asymptotically stable in
the absence of diffusion. Thus, there are only Turing conditions \eqref{Turing cindition 1}, and \eqref{turing condition-M_model}.
According to these conditions, we investigate variation of parameters in plane $(n_1,d),\; (n_1, \gamma)$.
Figure~\ref{fig:B and Lambda plots}, demonstrates that for the given data set, the bifurcation parameters of models \eqref{primary model} and \eqref{non-degenerated} are $d_c=2.7190$ and $\gamma_c=5.4035$ correspondingly.
\subsubsection{Influence of Intermediate Diffusion on the Stability of Enzyme–Substrate Spatial Organization in model \eqref{primary model}}
To investigate the influence of diffusion on pattern formation and aggregation, we examined how variations in the diffusion coefficients of both the substrate ($C_0$) and the intermediate ($C_1$) affect the emergence of spatial structures. Our analysis reveals that increasing either $D_0$ or $D_1$ leads to a progressive reduction of the Turing instability region. As diffusion becomes more efficient, the concentration gradients of both chemical species are rapidly smoothed out, preventing the formation of localized substrate-enriched zones. Consequently, the spatial structures (lobes or peaks) that would otherwise arise from enzyme–substrate feedback are flattened, and the system tends toward a spatially homogeneous steady state. This stabilizing effect of molecular mobility is clearly visible in Figure~\ref{fig:n1d-Di}, where the Turing domain diminishes with increasing $D_0$ and $D_1$.
\begin{figure}[h]
		\centering
		{\label{fig:(d,n1)-D0}\includegraphics[width=0.48\textwidth]{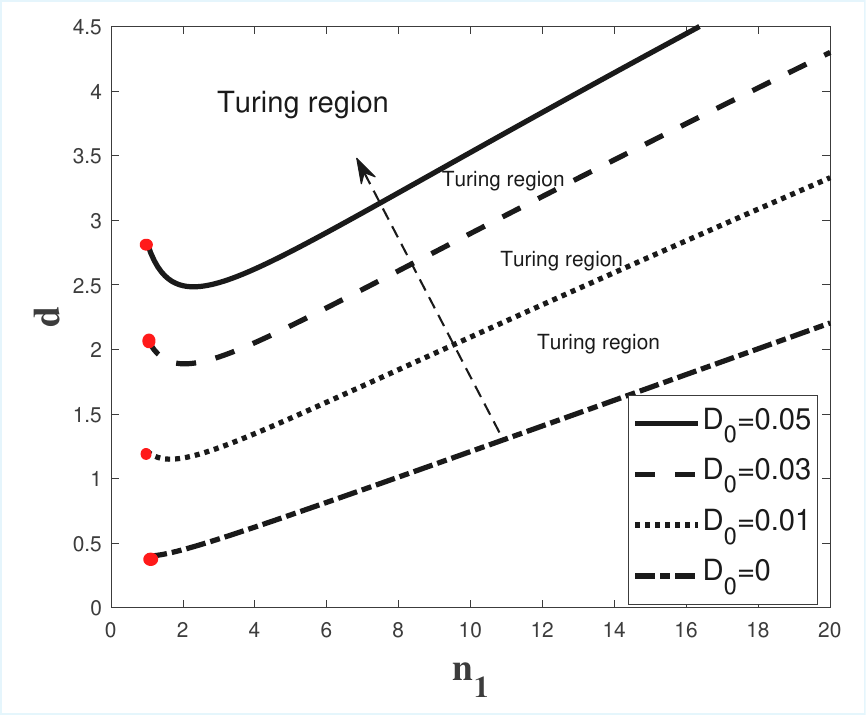}}
		\hspace{1mm}
		{\label{fig:fig:(d,n1)-D1}\includegraphics[width=.48\textwidth]{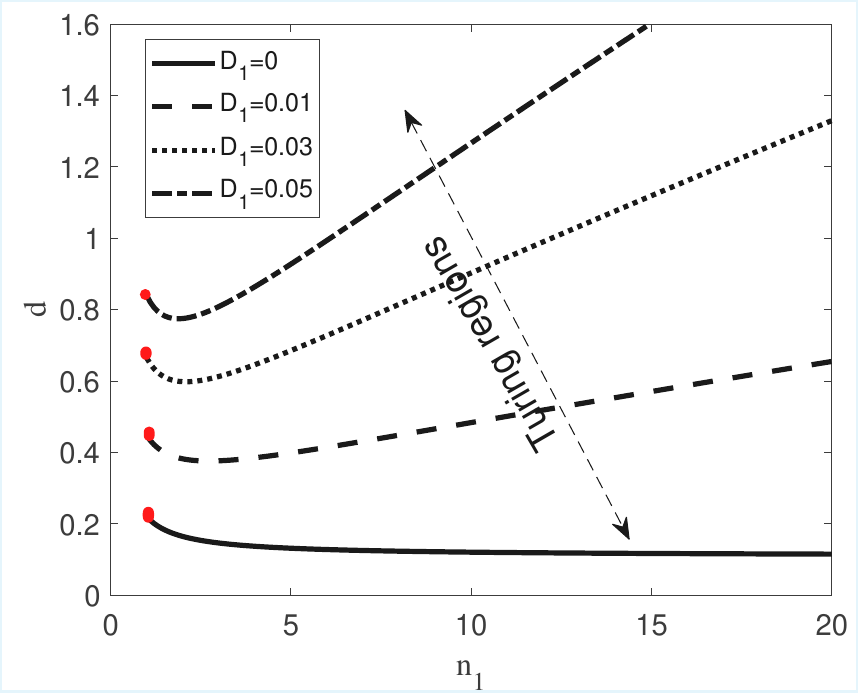}}
	\hspace{1mm}
		\caption{\textit{Variation of Turing regions in plane $(n_1, d)$  as the parameters $D_0$ (left plot) and $D_1$ (right plot) increase. Red dots in the plots indicate $d_c$ threshold for each given $D_0$ or $D_1$. Parameters have chosen like Figure~\ref{fig:B and Lambda plots}}.}
		\label{fig:n1d-Di}
		\hspace{1mm}
	\end{figure}
\subsubsection{Influence of Intermediate Diffusion on the Stability of Enzyme–Substrate Spatial Organization in model \eqref{non-degenerated}}
The analysis of Model~\eqref{non-degenerated} reveals that the diffusion of the intermediate plays a crucial role in regulating the onset of spatial organization. In Figure~\ref{fig:n1gamma-gammai}, as the diffusivity of the intermediate increases, the Turing region progressively shrinks and the instability threshold $\gamma_c(n_1)$ shifts toward higher values of the cross-diffusion parameter. This behavior indicates that enhanced intermediate mobility tends to equalize its spatial distribution, thereby weakening the gradients necessary to sustain localized substrate–intermediate interactions. From a biochemical perspective, this corresponds to a reduction in spatial compartmentalization: the mobile intermediate rapidly redistributes throughout the domain, preventing the formation of enzyme–substrate-enriched zones and driving the system toward a spatially homogeneous steady state.
\begin{figure}[H]
		\centering
		{\label{fig:(d,n1)-gamma0}\includegraphics[width=0.48\textwidth]{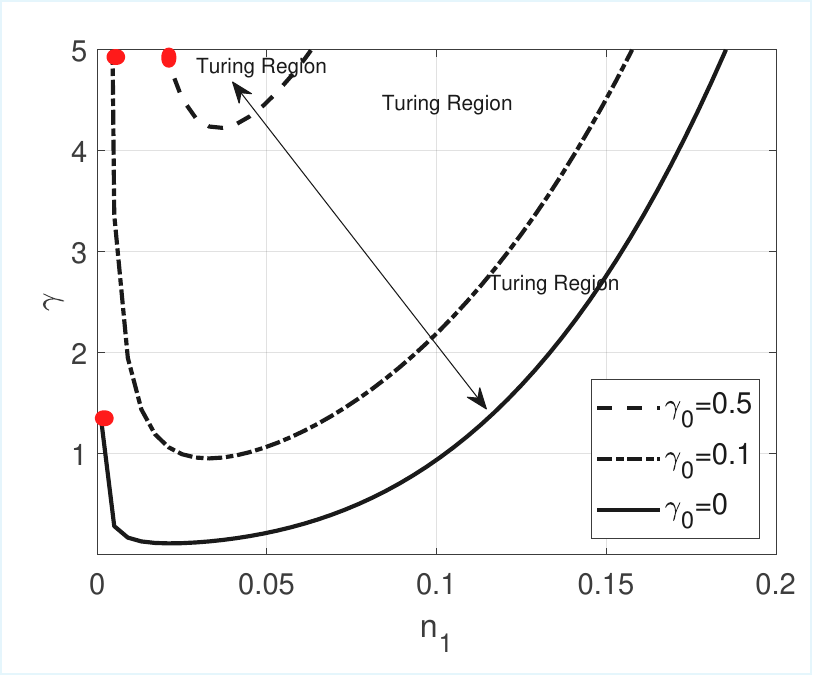}}
		\hspace{1mm}
		{\label{fig:fig:(d,n1)-gamma1}\includegraphics[width=.47\textwidth]{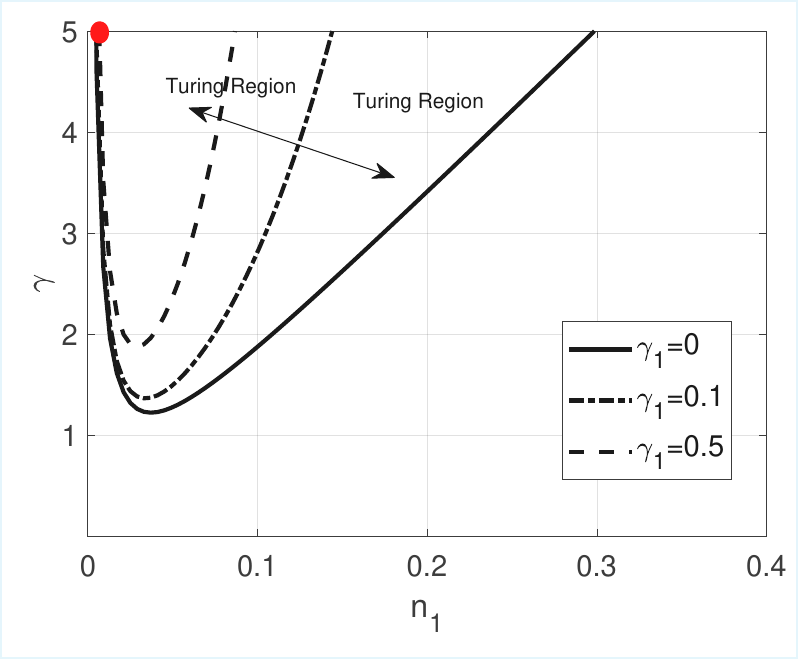}}
	\hspace{1mm}
		\caption{\textit{Variation of Turing regions in plane $(n_1, \gamma)$  as the parameters $\gamma_0$ (left plot) and $\gamma_1$ (right plot) increase. Red dots in the plots indicate $\gamma_c$ threshold for each given $\gamma_0$ or $\gamma_1$. Parameters have chosen like Figure~\ref{fig:B and Lambda plots}. } }
		\label{fig:n1gamma-gammai}
		\hspace{1mm}
	\end{figure}
	In both the primary and reduced models, increasing the substrate diffusivity $D_i, \gamma_i$ reduces the Turing region, indicating that higher substrate mobility stabilizes the homogeneous state. However, in the reduced model, where the cross-diffusion coupling depends explicitly on substrate concentration and reaction kinetics, this stabilizing influence is stronger. Thus, while both models exhibit diffusion-mediated suppression of pattern formation, the mechanistic coupling in the reduced model amplifies the homogenizing effect.
%
\vspace{-4cm}
\begin{figure}[H]
		\centering
\hspace{-1.5cm}
		\vspace{-3cm}
		{\label{fig:(d,n1)-kcat}\includegraphics[width=0.45\textwidth]{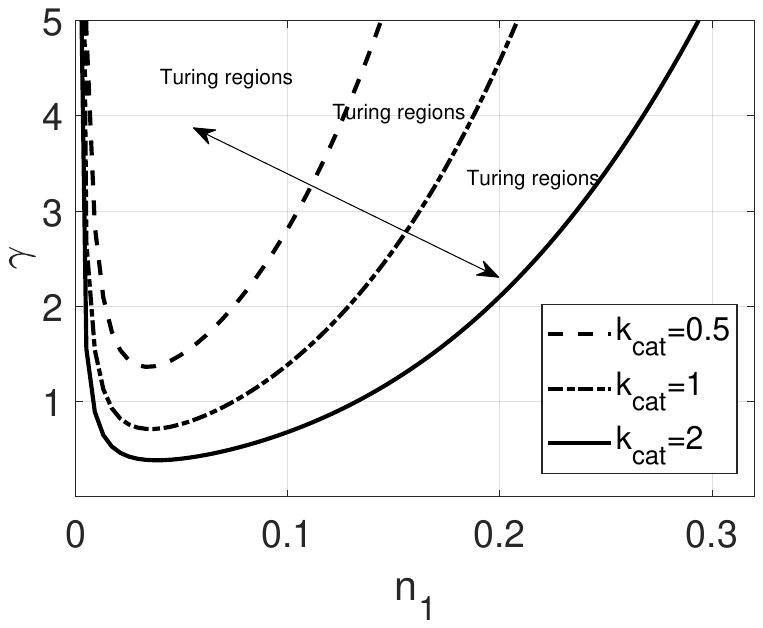}}
		\hspace{-2.5cm}
		{\label{fig:fig:(d,n1)-kd1}\includegraphics[width=.45\textwidth]{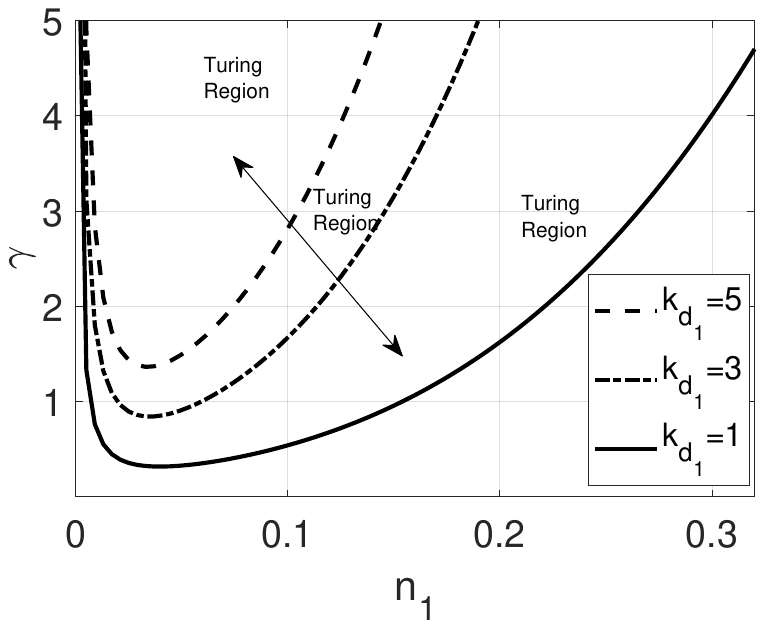}}
		\hspace{-2.5cm}
		{\label{fig:fig:(d,n1)-ka1}\includegraphics[width=.45\textwidth]{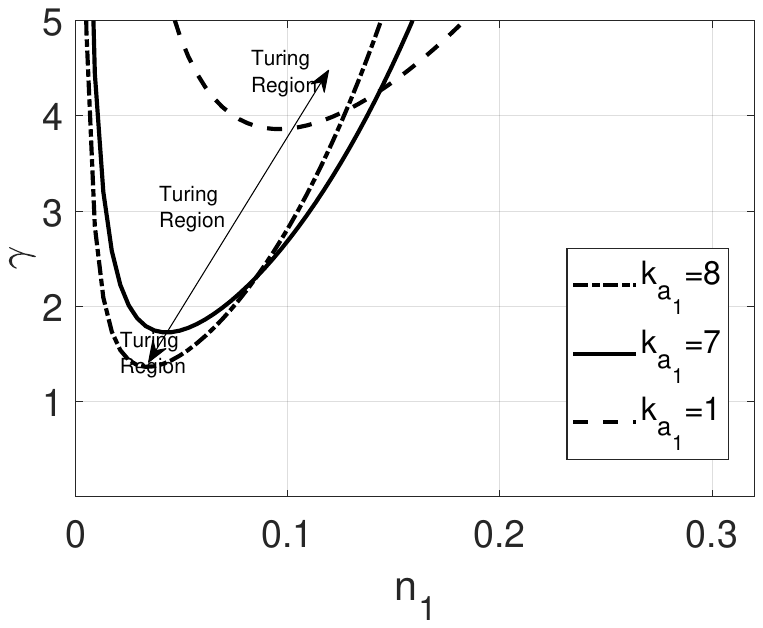}}
	\hspace{-1cm}
		\caption{\textit{Variation of Turing regions in plane  $(n_1,\gamma)$ as the parameter $k_{cat},\; k_{d_1}, \; k_{a_1}$ increases res. Parameters have chosen like Figure~\ref{fig:B and Lambda plots}.} }
		\label{fig:n1dgamma-kcat-kd1ka1}
		\hspace{1mm}
		\vspace{-2cm}
	\end{figure}	
\subsubsection{Influence of Kinetic Parameters on the Stability of Enzyme–Substrate Spatial Organization in Model~\eqref{non-degenerated}}  
Figures~\ref{fig:n1dgamma-kcat-kd1ka1} illustrate the influence of the kinetic parameters $k_{a_1}$, $k_{d_1}$, and $k_{cat}$ on the onset of diffusion-driven instability. In all cases, the system displays the characteristic U-shaped Turing boundary, but the extent and position of the unstable region are strongly affected by the reaction rates. Increasing the association rate $k_{a_1}$ enlarges the Turing region and shifts it toward higher $n_1$, indicating that stronger substrate--enzyme binding enhances the nonlinear coupling responsible for spatial self-organization. Conversely, increasing the dissociation rate $k_{d_1}$ diminishes the Turing domain, showing that faster complex dissociation stabilizes the homogeneous steady state. A similar expansion of the Turing region is observed for higher catalytic rates $k_{cat}$, where a more efficient turnover of the intermediate amplifies local concentration gradients and reinforces the enzyme--substrate feedback. Overall, strong association and high catalytic activity favor the emergence of spatial structures, whereas rapid dissociation weakens the coupling and suppresses pattern formation.
 \begin{figure}[h]
		\centering
		{\label{fig:(d,n1)-alpha0}\includegraphics[width=0.47\textwidth]{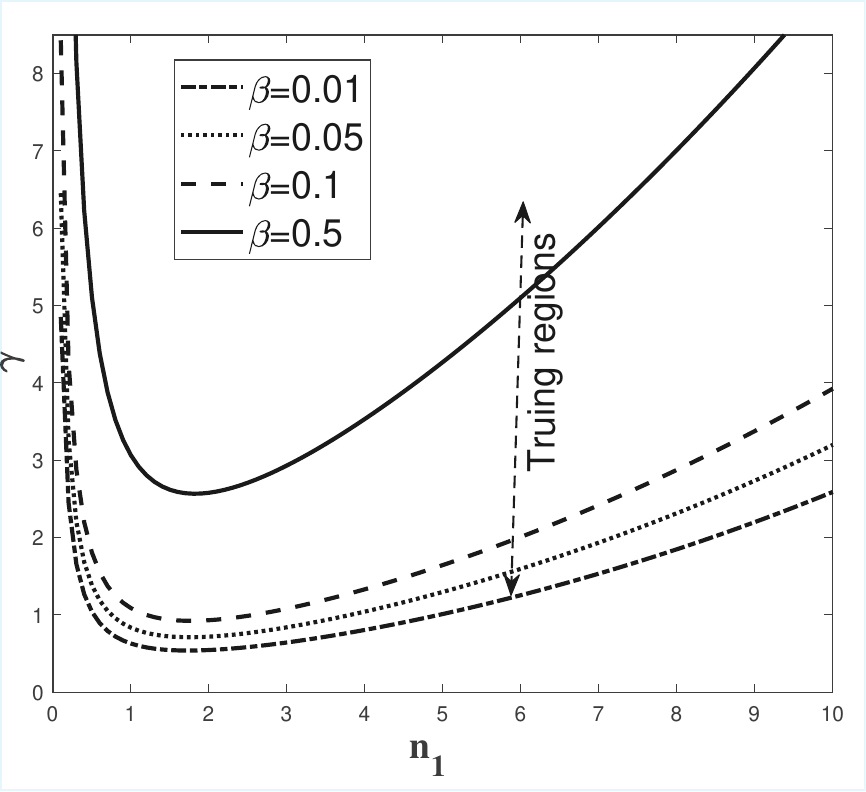}}
		\hspace{1mm}
		{\label{fig:fig:(d,n1)-beta}\includegraphics[width=.49\textwidth]{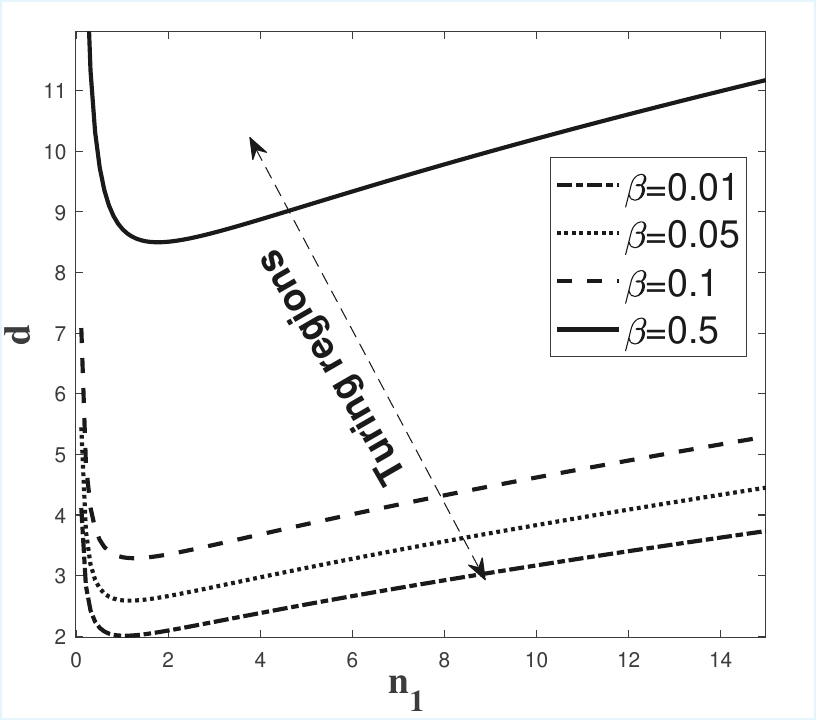}}
	\hspace{1mm}
		\caption{\textit{Variation of Turing regions in plane $(n_1, d)$ (right plot) and $(n_1,\gamma)$ (left plot) as the parameter $\beta$ increases. Parameters have chosen like Figure~\ref{fig:B and Lambda plots}.} }
		\label{fig:n1dgamma-beta}
		\hspace{1mm}
	\end{figure}
\begin{figure}[H]
		\centering
		{\label{fig:(d,n1)-alpha0}\includegraphics[width=0.48\textwidth]{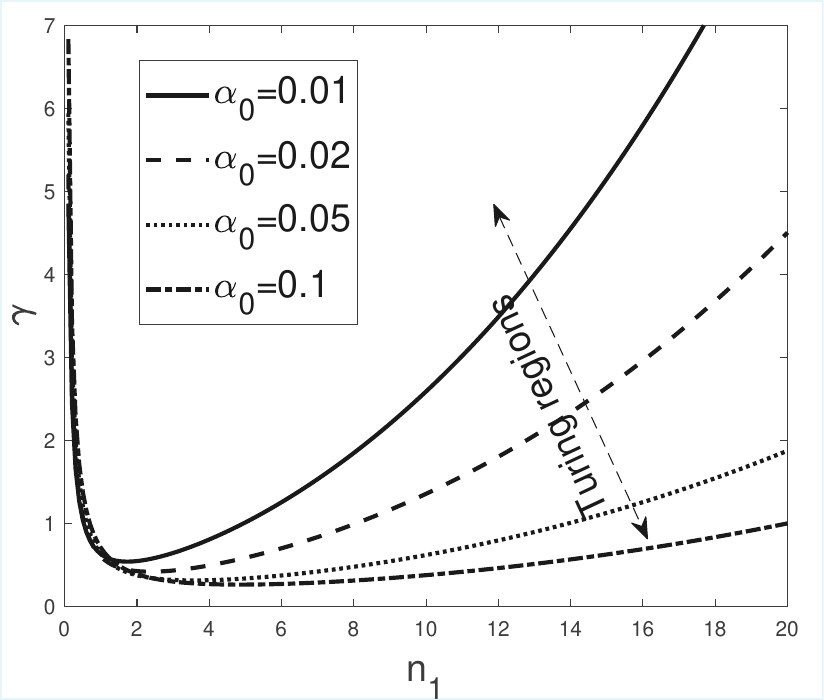}}
		\hspace{1mm}
		{\label{fig:fig:(d,n1)-beta}\includegraphics[width=.48\textwidth]{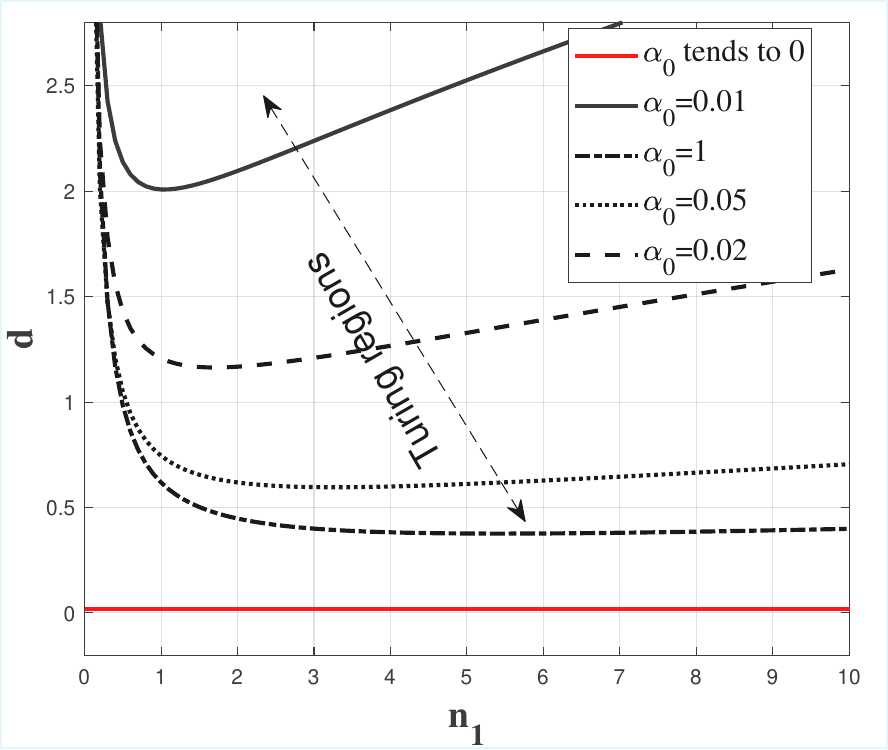}}
	\hspace{1mm}
		\caption{\textit{Variation of Turing regions in plane $(n_1, d)$ (right plot)
		 and $(n_1,\gamma)$ (left) as the parameter $\alpha$ increases. Parameters have chosen like Figure~\ref{fig:B and Lambda plots}. }}
		\label{fig:n1dgamma-alpha}
		\hspace{1mm}
	\end{figure}
	The  figures \ref{fig:n1dgamma-beta} and \ref{fig:n1dgamma-alpha} investigate how Turing instability regions shift under variations of key parameters, namely $\beta$ and $\alpha_0$, in the $(n_1, d)$ and $(n_1, \gamma)$ planes. In both the $(n_1, d)$ and $(n_1, \gamma)$ plots, increasing $\beta$ leads to higher critical values for $d$ and $\gamma$, respectively. This indicates that larger degradation or reaction rates of the inhibitor suppress pattern formation, requiring stronger diffusion or cross-diffusion to induce Turing instability. Conversely, the $(n_1, d)$ and $(n_1, \gamma)$ plots with varying $\alpha_0$ show that decreasing $\alpha_0$ significantly broadens the Turing region. In $(n_1, d)$  plane, when $\alpha_0 \rightarrow 0$, the threshold for $d$ approaches zero, suggesting that slower decay of the activator greatly facilitates the onset of instability.

In summary, the figures illustrate that the Turing space is highly sensitive to kinetic parameters. Smaller values of $\beta$ and $\alpha_0$ enlarge the admissible region in the parameter space where spatial patterns can form. This implies that systems with slower degradation or decay kinetics are more susceptible to Turing-type pattern formation. These insights are particularly relevant for enzymatic and reaction-diffusion models in biological systems, where parameter tuning can drive or suppress morphogenetic processes.	
	\begin{figure}[H]
		\centering
		{\label{fig:(d,n1)-alpha0}\includegraphics[width=0.47\textwidth]{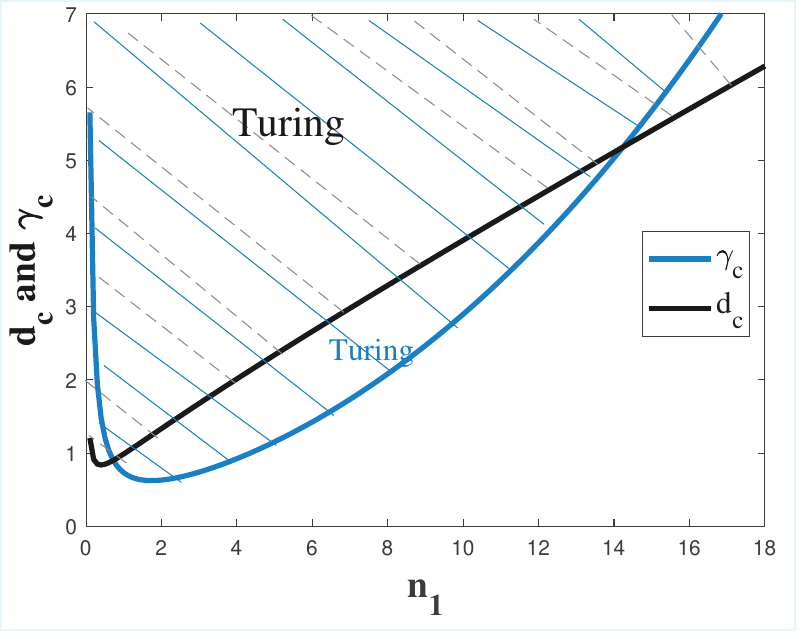}}
		\hspace{1mm}
		{\label{fig:fig:(d,n1)-beta}\includegraphics[width=.49\textwidth]{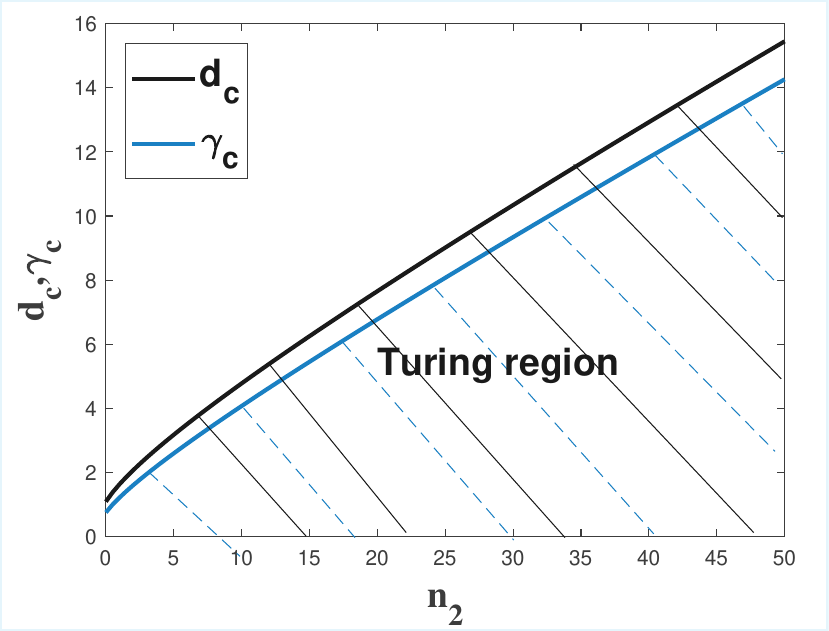}}
	\hspace{1mm}
		\caption{\textit{Variation of Turing regions in plane $(n_1, d)$ (right plot) and $(n_1,\gamma)$ (left plot). Variation of Turing regions in plane $(n_2, d)$ and $(n_2,\gamma)$ (right plot). Parameters have chosen like Figure~\ref{fig:B and Lambda plots}.}}
		\label{fig:n1n2.d.gamma}
		\hspace{1mm}
	\end{figure}	
	The plotted critical curves Figure~\ref{fig:n1n2.d.gamma}  highlight how the onset of pattern formation depends
nonlinearly on enzyme levels, with broader instability regions for low-to-moderate $n_1$, and increasingly suppressed Turing regions with high $n_2$. We analyze the critical diffusion parameters 
$d_c$ and $\gamma_c$
associated with Turing instability in two enzyme-based reaction-diffusion models. In the $n_1$-parameter space, both curves exhibit a non-monotonic, U-shaped behavior. This indicates that intermediate values of $n_1$ (the enzyme responsible for the first reaction step) are most conducive to pattern formation, requiring minimal diffusion or cross-diffusion to trigger instability. At both low and high 
$n_1$, the system becomes more stable, and higher values of $d_c$ and $\gamma_c$ are needed to induce patterns.
 
In contrast, the plots against $n_2$ show that both 
$d_c$ and $\gamma_c$ increase monotonically with enzyme concentration. This implies that the second enzyme, involved in downstream processes, stabilizes the system by diminishing feedback necessary for pattern formation. The widening gap between the two curves also suggests that the effect of cross-diffusion becomes increasingly significant as 
$n_2$ increases, reinforcing the role of downstream reactions in regulating spatial instabilities.
	\section{Weakly nonlinear Analysis of models}\label{sec:wnl}
	In this section we perform a weakly nonlinear analysis for the model \eqref{primary model} to predict the amplitude
of the pattern near the Turing threshold. The weakly nonlinear analysis
is based on the multiple scale methods (\cite{gambino2012turing}). The similar strategy has applied on the model \eqref{non-degenerated}.  Since near to the bifurcation the amplitude
of the pattern (diffusion-driven instability) has slow temporal scale, then a new temporal scale is defined.

The solution of the original system is written as a weakly nonlinear expansion
in the small control parameter $\varepsilon$.
We choose $\varepsilon^2 =\dfrac{d-d_c}{d_c}.$

The slow scale is obtained from the linear analysis: it is easy to prove that
$\lambda\sim \varepsilon^2$ and, since the growth rate of the perturbation is proportional to the
$exp(\lambda t)$, the slow time scale $T$ is order $\varepsilon^2$. Therefore, close to the threshold
we separate the fast time $t$ and slow time $T = \varepsilon^2t$, so that time derivative is
obtained as $\partial_t \rightarrow \partial_t + \varepsilon^2\partial_T$.

We separate the linear part from the nonlinear part:
\begin{equation}\label{linanonlin}
\partial_{t}\bm{w}=\mathcal{L}^{d_c}\bm{w}+ \mathcal{Q}_{D}^{d_c}(\bm{w},\bm{\nabla^2 w}), \qquad \bm{w}=[C_0-C_0^E, C_1-C_1^E],
\end{equation}
and linear operator is defined as $\mathcal{L}^{d_c} = \Gamma J+ \mathcal{D}^{d_c}\nabla ^2, $
where $\mathcal{D}$ and ${J}$ have defined in (\ref{reac-prim} and \ref{diffusion-prim}) and  nonlinear operators $ \mathcal{Q}_D^{d_c}(\bm{x},\bm{y})$ are introduced as:
\textit{$ \bm{x}=(x_{C_0},x_{C_1}),\; \bm{y}=(y_{C_0},y_{C_1})$},
 \begin{equation}
\mathcal{Q}_{D}(\bm{x},\bm{y})=\begin{bmatrix}
-dx_{c_0}y_{c_1}\\
0
\end{bmatrix},
\end{equation}
And moreover the bifurcation parameter and the solution are expanded asymptotically 
\begin{equation}\label{d-exp}
d=d_c+\varepsilon^2 d^{(2)}+O(\varepsilon^4),
\end{equation}
\begin{equation}\label{w-exp}
\bm{w}=\varepsilon \bm{w}_1+\varepsilon^2 \bm{w}_{2}+\varepsilon^3 \bm{w}_3+O(\varepsilon^4),
\end{equation}
Therefore, the diffusion matrix is given in terms of perturbation parameters as:
\begin{equation}\label{D-exp}
\mathcal{D}=\begin{bmatrix}
D_0  & -d_c C_0^E\\
 0 & D_1
\end{bmatrix}
+\varepsilon^2 \begin{bmatrix}
0 & -d^{(2)}C_0^E\\ 0 & 0\end{bmatrix}+O(\varepsilon^4)
\end{equation}
  and consequently 
  \begin{align}\label{QD-exp}
  \mathcal{Q}_D(\bm{w}, \nabla^2 \bm{w}) &= \varepsilon^2 d_c \mathcal{Q}_D(\bm{w}_1,\nabla^2\bm{w}_1)+\varepsilon^3 d_c\lbrace\mathcal{Q}_D(\bm{w}_1,\nabla^2\bm{w}_2)\\
  &+\mathcal{Q}_D(\bm{w}_2, \nabla^2\bm{w}_1)\rbrace+\varepsilon^4 d^{(2)}\mathcal{Q}_D(\bm{w}_2, \nabla^2 \bm{w}_2)+O(\varepsilon^5),
   \end{align}
  Now we replace all expansion in \eqref{d-exp}, \eqref{D-exp},\eqref{w-exp}, and \eqref{QD-exp} and sort according to the order of $\varepsilon $.
 At $O(\varepsilon)$,we obtain the following linear problem:
\begin{equation}
\mathcal{L}^{d_c}\bm{w}_1=0,\quad  \bm{w}_1=A(T){\rho} \cos(k_c x),
\end{equation}
such that satisfying the homogeneous boundary conditions
 where ${\rho}$ belongs to the $ker( J-k^2_c \mathcal{D}^{d_c})$.
 In this stage the $A(T)$, the amplitude of the pattern is arbitrary and the vector $\rho$ is considered constant such whose normalization is 
\begin{equation}
{\rho}=\begin{bmatrix}1\\ \rho_2\end{bmatrix},
\end{equation}
where $\rho_2=-\dfrac{-\alpha_0-k_1n_1-k_c^2D_0}{dC_0^Ek_c^2}$.

 Moreover, at $O(\varepsilon^2)$ there is this linear equation which must be solved:
 $$\mathcal{L}^{d_c} \textbf{w}_2=\bm{F},$$
According to the Fredholm alternative theorem, this equation has a  solution if and only if
 $\langle \bm{F} , \psi \rangle=0,$
  where $  \psi^{*}$
   is defined at $$ \psi= \begin{bmatrix}1\\R_1^*\end{bmatrix}\cos(k_c x)= \psi^* \cos(k_c x),$$
   and $ \langle .,.\rangle$ implied the scalar product in $\bm{L}^2(0,\dfrac{2\pi}{k^c})$ and $  \psi^* \in Ker( J-k^2_cD^{d_c})^{\dagger}$ where $\dagger$  shows transpose of complex conjugate of the matrix. 
   
In particular, $ \bm{F}=-\dfrac{1}{4}A^2\sum_{i=0,2}\mathcal{M}_i( \rho,( \rho)\cos(ik_c x)$,
in which $\mathcal{M}_i( \rho, \rho)=-i^2k_c^2\mathcal{Q}_D^{d_c}( \rho,\rho)$. Hence, the vector  $ \bm{w}_2$ is defined as \\
$\bm{w}_2=A^2(\bm{w}_{20}+\bm{w}_{22}\cos(2k_c x))$
 so that $\mathcal{L}_i^{d_c}\bm{w}_{2i}=-\dfrac{1}{4}\mathcal{M}_i( \rho, \rho),\; i={0,2}$ and $\mathcal{L}_{i}= -i^2k^2_c D^{d_c}$.\\
 In following, at $O( \varepsilon^3)$ we obtain the linear problem
 \begin{equation}{\label{order3}}
 \mathcal{L}^{d_c} \bm{w}_3=\bm{G},
 \end{equation}
 where 
 $
\bm{G}=\left(\dfrac{dA}{dT}\rho+A\bm{G}_1^{(1)}+A^3\bm{G}_1^{(3)}\right)\cos(k_c x)+A^3\bm{G}_3\cos(3k_c x),
$
in which \\
$ \bm{G}_1^{(1)}=k^2_c \begin{bmatrix} 0&d^{(2)}C_0^E\\0&0\end{bmatrix} \rho,$\\
$\bm{G}_1^{(3)}=-k_c^2\mathcal{Q}_D(\rho,\bm{w}_{22})-\dfrac{1}{2}k_c^2\mathcal{Q}_D( \bm{w}_{22},\rho),$\\
$\bm{G}_{3}=3k^2_c \mathcal{Q}_D^{d_c}( \rho,\bm{w}_{22}),$\\
Finally, by applying the solvability condition, we obtain the Stuart-Landau equation for the amplitude
\begin{equation*}
\dfrac{\partial A}{\partial T}=\sigma A-LA^3,
\end{equation*}
In addition, solvability of the equation $\eqref{order3}$ depends on $\langle \bm{G}, \psi \rangle=0$, whose coefficients  are given by\\
$$\sigma=-\dfrac{\langle \bm{G}_1^{(1)}, \psi\rangle}{\langle \rho, \psi\rangle},\; \; \bm{L}=\dfrac{\langle \bm{G}_1^{(3)}, \psi\rangle}{\langle \rho, \psi\rangle}, $$.

In the Stuart-Landau equation the coefficient $\sigma$ is always positive while L could be either negative or positive, corresponding to a subcritical or supercritical bifurcation.

The nontrivial solution of the amplitude equation is
\begin{equation}\label{L>0}
\bm{A}_{\infty}= \sqrt{\dfrac{\sigma}{L}},
\end{equation} 
which requires $L > 0$, and therefore the result of this analysis is valid only for supercritical bifurcations.
 
 Therefore, the asymptotic behavior of the solution is given by weakly nonlinear analysis of $O(\varepsilon^3)$ is:
 \begin{equation}\label{O3}
 \bm{W}=\varepsilon A \rho cos(k_cx)+\varepsilon^2A^2 [ \bm{w}_{20}+\bm{w}_{22}cos(2k_cx)]+O(\varepsilon^3). 
 \end{equation}
 \section{Numerical Results}\label{sec:num-results}
In this section, we exhibit numerical results. 
\subsection{Numerical methods}

The reaction-diffusion system was investigated numerically using a Galerkin spectral approximation. The spatial dependence of the variables was expanded in a truncated cosine Fourier series on the domain $x\in[0,2\pi/k]$, retaining four modes for each variable. Substituting these expansions into the governing equations and projecting onto the Fourier basis yields a system of eight ordinary differential equations for the modal amplitudes. The resulting system was integrated in time using MATLAB's \texttt{ode45} solver, which implements an adaptive Runge-Kutta (4,5) method with relative tolerance $10^{-6}$ and absolute tolerance $10^{-9}$.

 	\begin{figure}[h]
		\centering
		{\label{fig:Turing Pattern Model 1 for C_0}\includegraphics[width=0.48\textwidth]{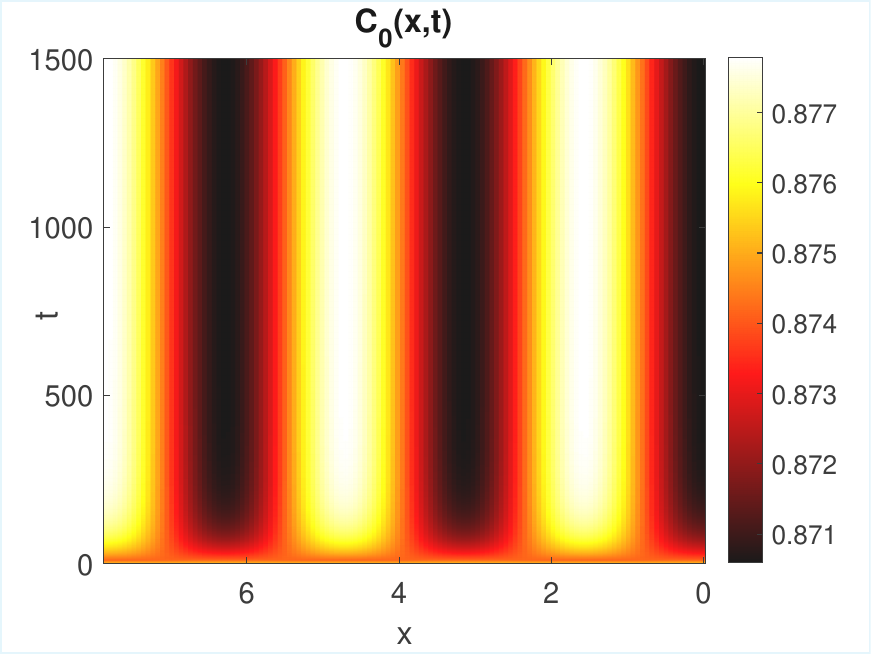}}
		\hspace{1mm}
		{\label{fig:fig:Turing Pattern Model 1 for C_1}\includegraphics[width=.48\textwidth]{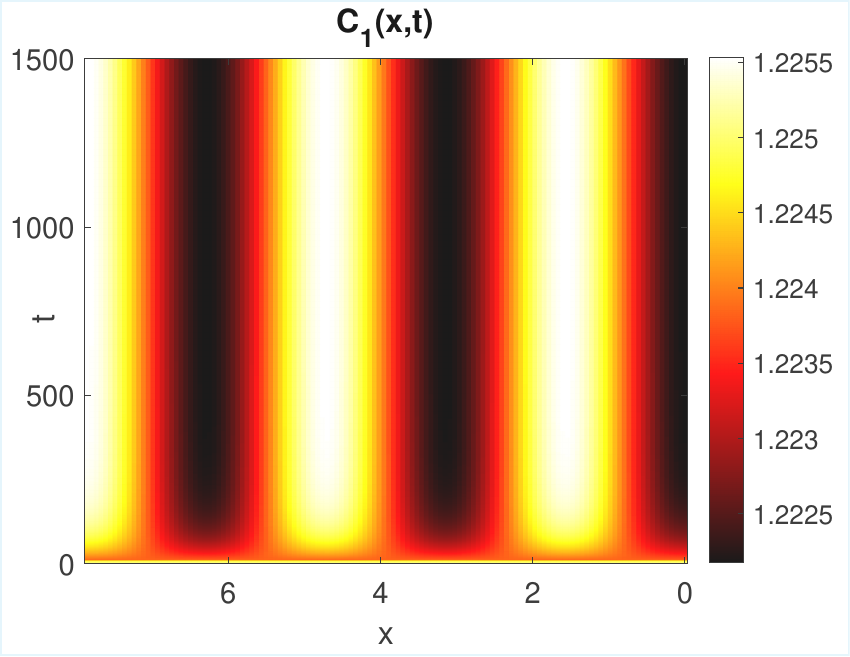}}
	\hspace{1mm}
		\caption{ \textit{Turing Pattern for given parameters: $n_1=3.5,\; n_2=1,\; k_1=k_2=2,\; \alpha_0=1,\; C_0^*=7,\; \beta=3,\; D_0=2,\; D_1=2.5$  that provides $C_0^E=0.87,\; C_1^E=1.22, \; k_c=2.8,\; d_c=9.51$.} }
		\label{fig:Turing Pattern Model 1}
		\hspace{1mm}
	\end{figure}
	\begin{figure}[h]
		\centering
		{\label{fig:Turing Pattern Model 1 for u_0}\includegraphics[width=0.48\textwidth]{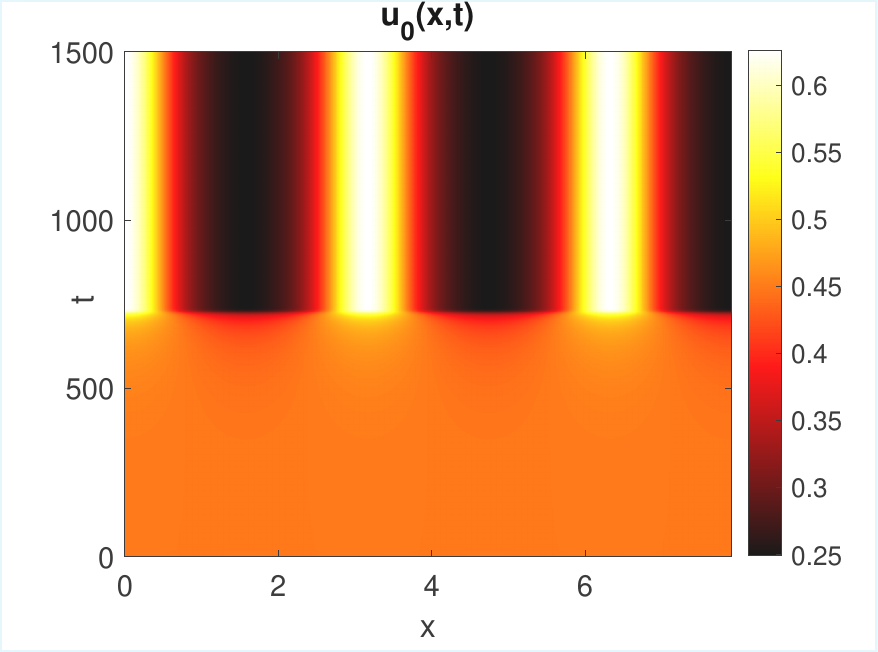}}
		\hspace{1mm}
		{\label{fig:fig:Turing Pattern Model 2 for u_1}\includegraphics[width=.48\textwidth]{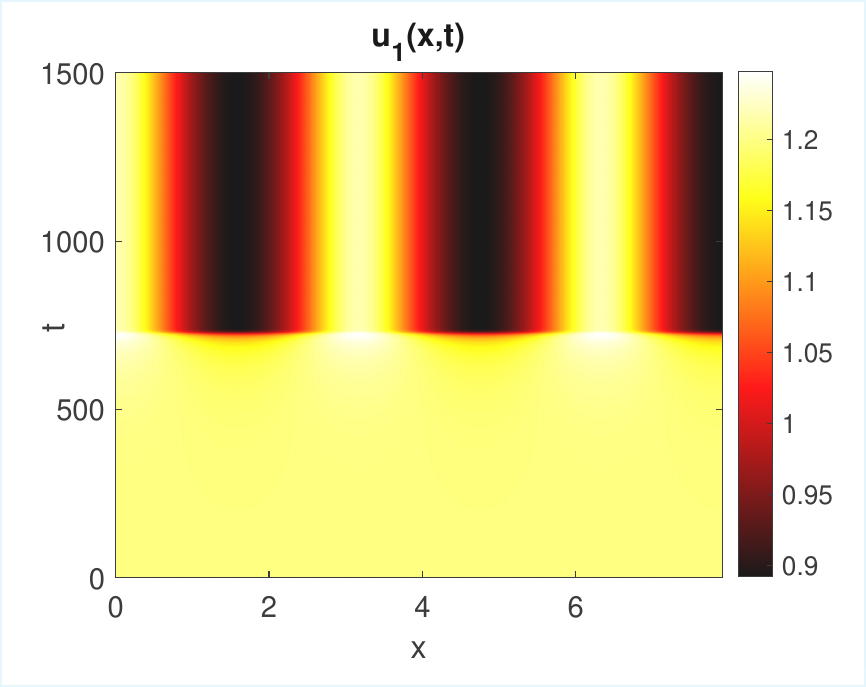}}
	\hspace{1mm}
		\caption{\textit{Turing Pattern for given parameters: $n_1=3.5,\; n_2=1,\; k_{a_1}=5,\; k_{cat}=8,\;k_{d_1}=2.5; k_2=2,\; \alpha_0=1,\; C_0^*=7,\; \beta=3,\; D_0=2,\; D1=2.5$  that provides $u_0^E=0.44,\; u_1^E=1.19, \; k_c=2.5,\; \gamma_c=14.76$.} }
		\label{fig:Turing Pattern Model 2}
		\hspace{1mm}
	\end{figure}
	In this formulation, diffusion terms appear as linear damping contributions proportional to the squared wavenumber, while the cross-diffusion term $C_0\nabla^2 C_1$ generates nonlinear coupling between Fourier modes in the reduced ODE system. The cosine basis automatically satisfies homogeneous Neumann (zero-flux) boundary conditions at the domain boundaries. Numerical reliability was verified by varying the perturbation amplitudes and integration times and by increasing the number of retained Fourier modes, which produced no qualitative change in the observed dynamics. Spatial solutions were reconstructed from the modal amplitudes to visualize the spatiotemporal evolution of the system.
Simulations were performed using the dimensional parameters.

 The numerical results clearly show that the primary model (Model~1) and the 
reduced model (Model~3) exhibit qualitatively similar Turing patterns, 
but with important quantitative differences in both amplitude and relaxation 
dynamics. In Model~1, the substrate–intermediate feedback is direct: the production 
of the intermediate $C_1$ depends instantaneously on the local availability of 
substrate $C_0$, and the curvature-driven cross-diffusion term $-d C_0 \nabla^2 C_1$ 
reinforces this coupling by transporting substrate toward regions where $C_1$ 
accumulates. As a result, the effective positive feedback loop between $C_0$ and 
$C_1$ is strong, leading to patterns with relatively large amplitude and fast 
convergence toward the stationary profile (Figure~\ref{fig:Turing Pattern Model 1}).

		\begin{figure}[H]
		\centering
		{\label{fig:Turing Pattern Model 1 for c_0}\includegraphics[width=0.48\textwidth]{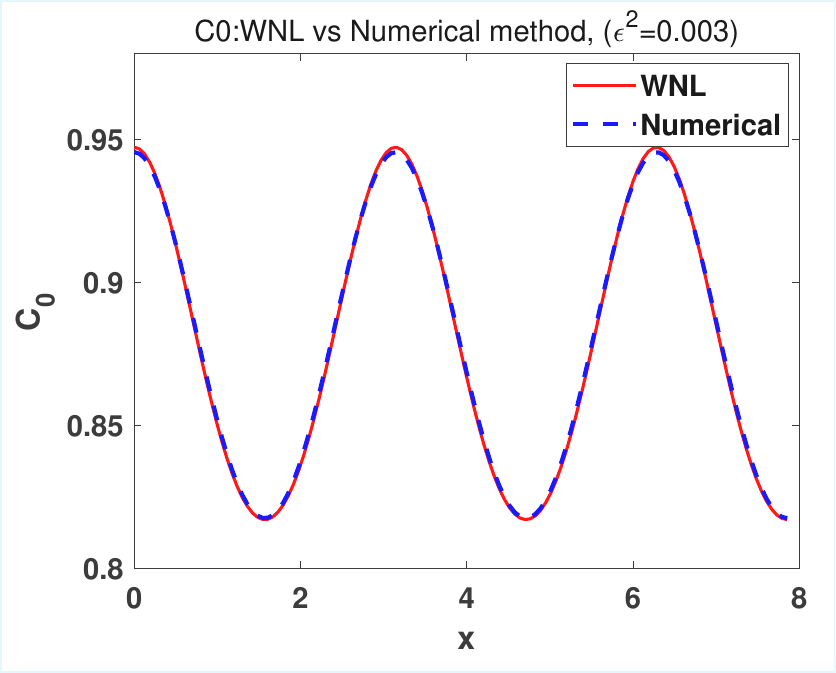}}
		\hspace{1mm}
		{\label{fig:fig:Turing Pattern Model 2 for c_1}\includegraphics[width=.48\textwidth]{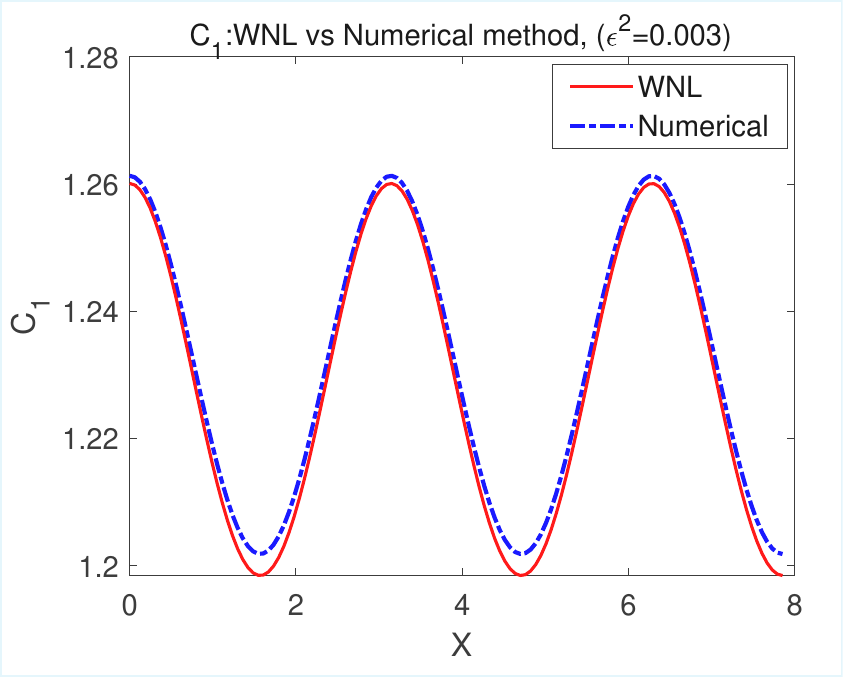}}
	\hspace{1mm}
		\caption{\textit{ Numerical results vs WNL of model 1 for given parameters: as Figure~\ref{fig:Turing Pattern Model 1}} and $\varepsilon^2=0.003$.  The WNL analysis shows s super critical Turing bifurcation since $\sigma$ and $L$ are both positive.  }
		\label{fig:Turing Pattern Model 1-WNLVsnum}
		\hspace{1mm}
	\end{figure}
	In contrast, the reduced model incorporates an explicit enzyme–substrate 
interaction, in which $S_0$ must first bind to the enzyme $E_1$ to form the complex 
$E_1S_0$ before being converted into the intermediate $S_1$. This reversible binding 
acts as a biochemical buffer: part of the substrate is temporarily stored in the 
complex, and the effective rate at which $S_0$ contributes to the production of $S_1$ 
is reduced to
$
\frac{k_{a_1} k_{\mathrm{cat}} n_1}{k_{d_1} + k_{\mathrm{cat}}} \; < \; k_1 n_1.
$
Consequently, the feedback between substrate and intermediate is weaker, and 
perturbations grow more slowly. This mechanistic buffering is visible in both the 
phase–plane trajectories and the numerical simulations: patterns in the 
reduced model have smaller amplitude, smoother spatial profiles, and longer 
relaxation times (Figure~\ref{fig:Turing Pattern Model 2}). 

Furthermore, the modified cross-diffusion structure of Model~3, where the coupling 
depends explicitly on the substrate concentration $u_0$, reduces the sensitivity of 
the system to spatial gradients. This further stabilizes the homogeneous state and 
explains why the Turing region of Model~3 is narrower and shifted toward higher 
diffusion or cross-diffusion values compared with Model~1.

The excellent agreement between numerical simulations and the weakly nonlinear 
predictions (Figures~\ref{fig:Turing Pattern Model 1-WNLVsnum}-\ref{fig:Turing Pattern Model 2-WNLVsnum}) confirms that these differences arise from fundamental 
kinetic mechanisms rather than numerical artifacts. Overall, the numerical results 
demonstrate that explicit enzyme–substrate binding attenuates the strength of the 
feedback necessary for Turing instability, leading to weaker and more finely regulated 
spatial patterns. Biologically, this means that enzymatic complex formation acts as a 
natural buffering process that prevents excessive accumulation of intermediates and 
ensures more stable metabolic organization.
		\begin{figure}[H]
		\centering
		{\label{fig:Turing Pattern Model 1 for u_0}\includegraphics[width=0.48\textwidth]{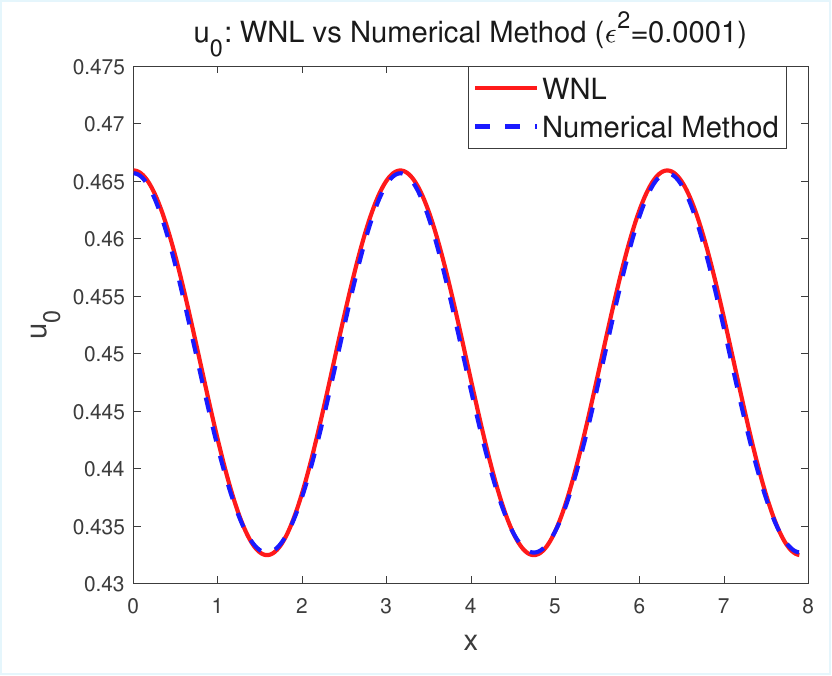}}
		\hspace{1mm}
		{\label{fig:fig:Turing Pattern Model 2 for u_1}\includegraphics[width=.48\textwidth]{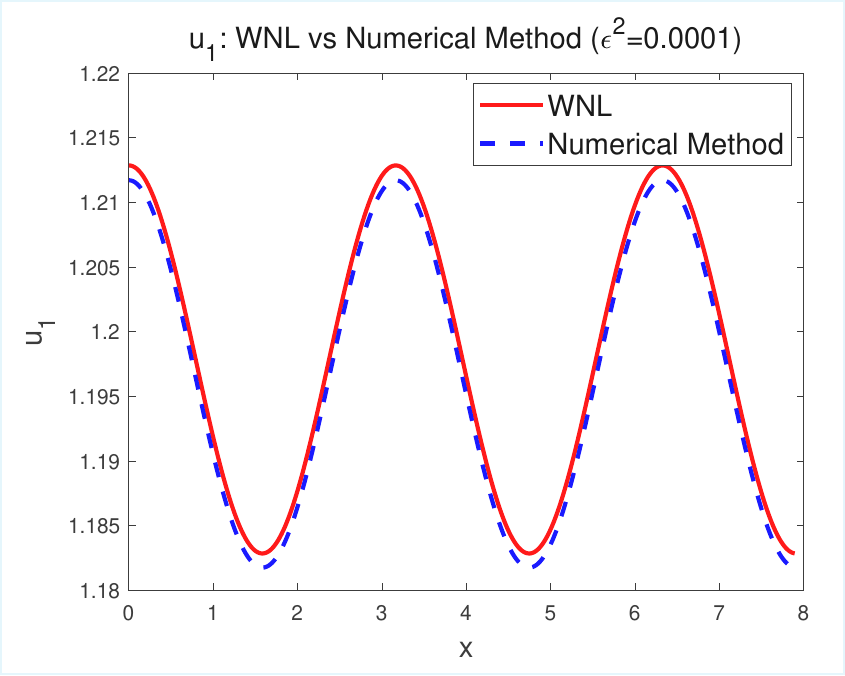}}
	\hspace{1mm}
		\caption{\textit{ Numerical results vs WNL of model 2 for given parameters: as Figure~\ref{fig:Turing Pattern Model 2}  and $\varepsilon^2=0.0001$. The WNL analysis shows s super critical Turing bifurcation since $\sigma$ and $L$ are both positive.  }}
		\label{fig:Turing Pattern Model 2-WNLVsnum}
		\hspace{1mm}
	\end{figure}		
\section{Conclusion}\label{sec:conclu}

This work examined how enzymatic reaction kinetics and diffusion-driven
transport interact to generate spatial metabolic organization. Linear
stability and Turing analyses showed that diffusion and curvature-driven
cross-diffusion can destabilize otherwise stable homogeneous steady states,
producing spatial patterns through local enzyme--substrate feedback.
Incorporating explicit enzyme--substrate complex formation introduces
additional kinetic structure that modifies reaction balance and effective
transport, thereby shifting the Turing threshold and altering the size and
location of the instability region. The resulting reduced 
formulation removes the structural degeneracy associated with a manifold of
equilibria and enables a complete weakly nonlinear (WNL) analysis while
preserving the essential nonlinear coupling between catalysis and spatial
organization.

Theoretical predictions from the WNL framework are strongly supported by
numerical simulations, which show excellent agreement in both pattern
amplitude and spatial profile. Comparison between the simplified and
reduced models reveals systematic kinetic and spatial differences:
explicit enzyme--substrate binding shifts the homogeneous equilibrium toward
higher intermediate concentrations, slows relaxation toward steady state,
and leads to spatial patterns of reduced amplitude. These results highlight
the buffering role of reversible complex formation and demonstrate how
enzymatic binding kinetics regulate the strength of the feedback loop driving
diffusion-driven instability.

Beyond their mathematical implications, these findings admit a natural
interpretation in the context of liquid--liquid phase separation (LLPS).
Within biomolecular condensates, enzymes and metabolites experience altered
mobility, molecular crowding, and transient binding interactions that jointly
modulate catalytic flux and spatial organization. The simplified model
captures an effective, coarse-grained description of enzymatic feedback,
whereas the reduced formulation reflects a buffered regime in which
reversible complex formation attenuates spatial feedback. The resulting
differences in pattern amplitude and stability mirror known biophysical roles
of LLPS in regulating metabolite gradients, enhancing reaction efficiency,
and preventing uncontrolled accumulation of intermediates.

Overall, this study establishes a quantitative link between enzymatic binding
kinetics, diffusion-driven instabilities, and mesoscale spatial organization
such as LLPS. The framework developed here can be extended to more complex
metabolic pathways, multi-enzyme assemblies, and other compartmentalized
biochemical systems, providing a versatile foundation for understanding how
molecular self-organization shapes cellular metabolism and morphogenesis.

The mechanistic structure of the full enzymatic model suggests rich reduced dynamics near criticality. 
	\section{Acknowledgments}
	 The author warmly thanks Prof. Stepan Timr (J. Heyrovský Institute of Physical Chemistry,
Czech Academy of Sciences)
 for his constant support, stimulating discussions, and many helpful suggestions that greatly improved the quality of this work. FF is also grateful for support from the National Group of Mathematical Physics (GNFM-INdAM).

\section*{Statements and Declarations}
\textbf{Competing interests:}  
The author declares that there are no competing interests.\\
\textbf{Data availability:}  
No new experimental data were generated in this study. All simulation codes and any numerical data supporting the findings are available from the corresponding author on reasonable request.\\
\textbf{Funding:}  
No external funding was received for this research.\\
\textbf{Ethics approval and consent to participate:}  
Not applicable.\\
\textbf{Authors’ contributions:}  
The author  developed the mathematical model, performed all analyses and simulations, and wrote the entire manuscript.

\bibliographystyle{plain}  
\bibliography{ReferencesEnzyme}  


\end{document}